\documentclass[a4paper,reqno]{amsart}

\usepackage[english]{babel} 
\usepackage{csquotes}
\usepackage[colorlinks=true,linkcolor=black,anchorcolor=black,citecolor=black,filecolor=black,menucolor=black,runcolor=black,urlcolor=black]{hyperref} 
\usepackage{graphicx,color,comment}  
\usepackage[dvipsnames]{xcolor} 
\usepackage{enumitem,marvosym} 
\usepackage{rotating} 
\usepackage[absolute]{textpos}

\usepackage{mdframed}
\usepackage{xfrac}
\usepackage{nicefrac}
\usepackage{comment,cancel}
\usepackage[bbgreekl]{mathbbol}

\usepackage[backend=biber, giveninits=true, style=alphabetic, sorting=nyt,isbn=false, maxalphanames=5,url=false, maxbibnames=99]{biblatex}

\usepackage{mathrsfs}
\usepackage[mathcal,mathscr]{euscript} 

\usepackage{amsmath,amssymb,amsfonts,amsthm} 
\usepackage{mathtools,stmaryrd}  
\definecolor{light-gray}{gray}{0.95}
\usepackage[color=light-gray]{todonotes}
\usepackage{tikz-cd}  
\usepackage[all,cmtip]{xy} 
\usepackage[bbgreekl]{mathbbol} 
\usepackage{mathrsfs,slashed}

\usepackage{thmtools}
\theoremstyle{definition}
\newtheorem{definition}{Definition}[subsection]

\newtheorem{theorem}[definition]{Theorem}
\newtheorem{conjecture}[definition]{Conjecture}
\newtheorem{proposition}[definition]{Proposition}

\newtheorem{lemma}[definition]{Lemma}

\newtheorem{corollary}[definition]{Corollary}
\declaretheoremstyle[
  spaceabove=5pt, spacebelow=5pt,
  headfont=\itshape,
  notefont=\normalfont, notebraces={(}{)},
  bodyfont=\normalfont,
  postheadspace=1em,
  qed=$\spadesuit$
]{pluto2}
    \declaretheorem[style=pluto2,name=Remark,    sibling=definition]{remark}

\renewcommand{\top}{\mathrm{top}}
\newcommand{\fg}{\mathfrak{g}}
\newcommand{\calF}{\mathcal{F}}

\newcommand{\qsp}[2]{\,\ensuremath{\raise.5ex\hbox{$#1$}\big\slash\raise-.5ex\hbox{$#2$}}}

\newcommand{\LQ}{\mathcal{L}_{Q}}
\newcommand{\LE}{\mathcal{L}_{\mathsf{E}}}
\newcommand{\bD}{\mathbb{\Delta}}
\newcommand{\bbL}{\mathbb{L}}

\newcommand{\oloc}{\Omega_{\mathrm{loc}}}
\newcommand{\ohor}{\Omega_{\mathrm{H}}}
\newcommand{\oham}{\Omega_{\mathrm{ham}}}
\newcommand{\oht}{\Omega_{\mathrm{ht}}}

\newcommand{\Xloc}{\mathfrak{X}_{\loc}}
\newcommand{\fracF}{\mathfrak{F}}
\newcommand{\loc}{\mathrm{loc}}
\newcommand{\ham}{\mathrm{ham}}
\newcommand{\src}{\mathrm{src}}

\newcommand{\id}{\mathrm{id}}
\newcommand{\im}{\mathrm{Im}}

\newcommand{\bb}[1]{\mathbb{X}_{#1}}

\newcommand{\bom}{\bul{\omega}}
\newcommand{\bth}{\bul{\theta}}
\newcommand{\bF}{\bul{F}}
\newcommand{\bL}{\bul{L}}
\newcommand{\bbP}{\mathbb{P}}

\newcommand{\wt}[1]{\widetilde{#1}}

\newcommand{\Liou}[1][\bul{\eta}]{\mathsf{L}_{#1}}
\newcommand{\bul}[2][]{#2^{\bullet #1}}

\title{Homotopies for Lagrangian field theory}
\author{Michele Schiavina}
\author{Jonas Schnitzer}
\address{Dipartimento di Matematica, Università di Pavia, via Ferrata 5, 27100, Pavia, Italy}
\email{michele.schiavina@unipv.it, jonaschristoph.schnitzer@unipv.it}
\date{}

\begin{document}
\begin{abstract}
    Consider the variational bicomplex for $\mathcal{E}$ the space of sections of a graded, affine bundle. Local functionals $\mathcal{F}$ are defined as an equivalence class of density-valued functionals, which represent Lagrangian densities. A choice of a $k$-symplectic local form $\omega$ on $\mathcal{E}$ induces a Lie$[k]$ algebra structure on (Hamiltonian) local functionals $(\mathcal{F}_\ham,\{\cdot,\cdot\}_{\mathrm{ham}})$. For any $\omega$ and any choice of a cohomological vector field $Q$ compatible with $\omega$, we build three explicit $L_\infty$ algebras on a resolution of $\mathcal{F}_{\mathrm{ham}}$, which are all $L_\infty$ quasi-isomorphic to a dgL$[k]$a $(\mathcal{F}_{\mathrm{ham}},d_{\mathrm{ham}},\{\cdot,\cdot\}_{\mathrm{ham}})$. In particular, one of our equivalent $L_\infty$ algebras is a dgL$[k]$ algebra. In the case $k=-1$, this provides an explicit lift of the standard Batalin--Vilkovisky framework to local forms enriched by the $L_\infty$ structure, in terms of local homotopies, which interprets the modified classical master equation as a Maurer--Cartan equation for the distinguished dgL$[k]$a we construct. We conjecture that the data of a lift to local forms of a BV theory contains a homotopy moment map on the cohomology of the Koszul complex of the underlying Lagrangian field theory.
\end{abstract}
\maketitle

\section{Introduction}
One of the key features in the Lagrangian theory of fields is locality. This is the requirement that the basic structures upon which classical field theory is built should be determined at a point in some ambient space(-time), and observables should not depend on field theory data at arbitrary large separations. This axiom is successfully encoded by means of the variational bicomplex, a natural structure that arises when looking at jets of sections of fibre bundles over said spacetime manifold. This allows one to work with local forms on the space of sections of the given underlying bundle. From the point of view of the geometry of the problem, locality effectively tames the infinite dimensionality of the spaces involved.

Further mathematical depth and complexity emerges when looking at field theory data that is invariant under local Lie group actions. This feature is often called \emph{gauge symmetry} and is both a resource of important insight for pure mathematics as well as a challenge, since ``physical observables'' should only depend on gauge-equivalence classes of configurations, and thus are expected not to depend on choices of representatives. When formalizing a local theory of classical fields with local symmetry (especially in view of its quantisation), one needs to combine the variational bicomplex with a method to account for invariants. 

This task is handled via the Batalin--(Fradkin)--Vilkovisky framework \cite{BFV0,BFV1,BV1,BV2,BFV2}, where out of the classical field theory input one builds a differential graded-symplectic manifold whose cohomology describes the locus of classical solutions of the equations of motion of the given theory, modulo the gauge action. In other words, it aims for a cohomological description of the moduli space of the theory (see \cite{CMS_BV} and references therein).

One important advantage of this method is that the dg manifold one builds out of a gauge theory is again phrased in terms of local field theory data: one then works with a \emph{graded} local Lagrangian field theory, whose associated local cohomology theory is the original physical space of interest. Moreover, it was suggested that, in order to fully exploit the power of the BV formalism, one should probe what the theory induces on lower dimensional strata. This led to the BV-BFV formalism\footnote{The second acronym is after Batalin--Fradkin--Vilkovisky.} of Cattaneo, Mnev and Reshetikhin \cite{CMR2014,CMR18}. 

While the algebraic structure underpinning the BV formalism for closed manifolds\footnote{Or noncompact manifolds with appropriate compact support conditions, see e.g.\ \cite{rejznerbook}.} is well understood, much less clear is how one should \emph{lift} this structure to the variational bicomplex. For a review on modern techniques on this topic see \cite{Grigoriev} and references therein. Indeed, one can define a shifted Poisson bracket on \emph{local (Hamiltonian) functionals}, denoted by $\calF_{(\ham)}$ and defined as equivalence classes of Lagrangian densities, seen as a subalgebra of $C^\infty(\mathcal{E})$ by integration of the Lagrangian density (see Definition \ref{def:localfunctionals} and Remark \ref{Rem: Stasheff}). The BV data is encoded by a compatible differential, which turns $\calF_\ham$ into a dgLa $(\calF_\ham,d_\ham,\{\cdot,\cdot\}_\ham)$.

The variational bicomplex was also used in \cite{BFLS} to construct a strong--homotopy Lie (a.k.a.\ $L_\infty$) algebra covering the Lie algebra of Hamiltonian local functionals, when symplectic data is given. One limit of the mentioned approach is that it assumes the existence of a resolution of the space of local functionals, which was explicitly shown to exist only for $M$ contractible. Moreover, the mentioned work fails to provide explicit expressions for the higher bracket of the resulting $L_\infty$ algebra.

In this paper we improve on that construction by providing, given a choice of homotopies for the variational bicomplex:
\begin{enumerate}
    \item A resolution of the space of local functionals $\calF$ for any (affine) bundle $E\to M$. We do this by explicitly constructing a deformation retract (Theorem \ref{thm:resolution}, see also Theorem \ref{Thm: DefRetractwithoutconstants}).
    \item An $L_\infty$ algebra, whenever $\mathcal{E}$ is endowed with a (weak) $k$-symplectic structure $\omega$, determined by a \emph{choice of theory} in the form of a local, cohomological vector field $Q\in\Xloc(\mathcal{E})$ (Theorem \ref{Thm: Linfty}). Such $L_\infty$ data is $L_\infty$ quasi-isomorphic to the dg Lie algebra of Hamiltonian functionals $(\calF_\ham,d_\ham,\{\cdot,\cdot\}_\ham)$, and it can be thought of as \emph{adapted} to the choice of a (BV) theory. This is done by promoting the deformation retract presented above to a $L_\infty$ quasi-isomorphism, for which an explicit quasi-inverse is provided (Lemma \ref{lem:quasiinverse}).
    \item An $L_\infty$ algebra that only depends on $(\mathcal{E},\omega)$, and is also $L_\infty$ quasi-isomorphic to Hamiltonian functionals (Proposition \ref{Prop: twisteq}). This is shown to be related to the above ($Q$-dependent) $L_\infty$ algebra \emph{via} twist.
    \item A dgL$[k]$ algebra which is $L_\infty$ 
    quasi-isomorphic to all of the above (Eq. \ref{e:Stasheffsres}).
    \item Explicit formulas for all $L_\infty$ morphisms above, and their quasi-inverses, for given choices of homotopies.
    \item A reformulation of the BV formalism in terms of the above constructions. Indeed, we show how the lift of the BV formalism to the variational bicomplex is encoded by the notion of Hamiltonian triples (Definition \ref{def:hamtriples}), and how natural constructions within the formalism provide us with a distinguished Maurer--Cartan element of the associated $L_\infty$ algebras (Theorems \ref{thm:Hamtriples} and \ref{thm:MCelement}). This recovers and interprets the approaches presented in \cite{CMR2014} and \cite{MSW}.
    \item We show how said lift of a BV theory within the variational bicomplex in terms of Hamiltonian triples is tied to the existence primitives of an extended BV multisymplectic form (Theorem \ref{thm:multisympBV} and Corollary \ref{cor:twistedmultisymplectic}). In Conjecture \ref{Conjecture} we propose that such data should yield a homotopy moment map on the solution space of the underlying Lagrangian field theory \cite{CalliesFregierRogersZambon,Bernardy_homotopy}. 
\end{enumerate}

It should be observed that several people worked on  questions related to the one we address here: Beyond the already mentioned \cite{BFLS}, we chiefly refer to \cite{BarnichBrandtHenneaux,BarnichHenneaux} among others. Our results are an improvement on the current state of the art. First, we close a gap in generality: the cone construction we implement in several places is a simple way to extend even the implicit results of \cite{BFLS} beyond contractible base manifolds. Moreover, one of the most important features that distinguishes our approach from the previous ones is that we do not use abstract results to conclude that certain (infinity)-structures simply exist by homotopy-transfer, but we rather use homotopies to give explicit formulas and control the induced structures. 

This is important because it allows us to \emph{classify} the choices involved in selecting an explicit presentation of the $L_\infty$ algebra underlying the Lagrangian field theory in the BV formalism (Theorem \ref{thm:Hamtriples}). In particular, since such choices are controlled by \emph{Hamiltonian triples}, we encode the data change in the notion of a triple redefinition (Definitions \ref{def:hamtriples} and \ref{def:tripleredef}).

From the point of view of physics, we note that the ``correct'' definition of Noether charges for Lagrangian field theory, which then leads to measurable physical quantities, is nontrivial and potentially ambiguous when phrased in standard language (see, only to name a few, \cite{FOPS, OdakRBSpeziale, AshtekarSpeziale_charges, spezialelectures} and the precursor of this analysis \cite{Wald_Zoupas}). This is because such definition requires the choice of primitives of the horizontal complex, and these choices are usually done by hand, and not systematically. Our results clarify how they instead can be made to descend from a single choice of homotopy, and are thus valid at all codimensions. This ultimately leads to the physically-relevant statement that said choices are ultimately immaterial. We believe this will have applications in the analysis of field-theoretic data on manifolds with boundaries and corners, see e.g.\ \cite{RielloSchiavina_PS}.

Furthermore, the next step in this program is to extend our considerations to the case in which a compatible associative product is given on the space of local functionals. In other words, one would like to discuss the local Poisson-infinity structure on functionals, along the lines of \cite{GwilliamRejzner_Obs}. This is necessary when thinking about some ``classical-to-quantum'' map of $L_\infty$ or even $P_\infty$ algebras (see. e.g.\ \cite{Pridham}).  In order to do this one has to define multilocal functionals, and establish the existence of Poisson data over them. This is technically  complicated by a number of challenges, both of algebraic nature (see \cite{FrabettiKravchenkoRyvkin}) as well as analytical \cite{HwkinsRejznerVisser_equicausal}. Moreover, deformation quantisation in infinite dimensions is a subtle problem, where only a few examples are known that rely on explicit formulas. In this sense our explicit, homotopy dependent, formulas, may serve as a technical tool in this regard.

We conclude this note by relating the BV formalism in the variational bicomplex to multisymplectic geometry \cite{gotay1998momentum,gotay2004momentum,rogers,BlohmannLFT}. 

Some of the structures that appear naturally in our investigations were also discovered in relation to characteristic classes for dg manifolds by \cite{KotovStrobl}. We note that, while some of the ``extended BV structures'' appearing in this context emerge naturally in the context of AKSZ theories, our analysis is general. One could argue that a large class of gauge theories can be given an AKSZ formulation, for instance following \cite{Grigoriev_parentLag}, this is however outside of the scope of our paper.

This manuscript is structured as follows: after a brief section on preliminaries about the variational bicomplex, homological algebra and $L_\infty$ algebras, we tackle the homological algebra of the variational bicomplex by establishing a number of useful deformation retracts in Section \ref{sec:homVBC}. These are then used in Section \ref{sec:Linfty} to build various explicit $L_\infty$ algebras and relations among them on the variational bicomplex. We then reformulate the BV formalism in these terms in Section \ref{sec:BVBFV}.

\section*{Acknowledgements}
We thank L.\ Vitagliano for providing comments on a draft and E.\ Getzler for insightful exchange. M.S.\ also thanks C.\ Blohmann and J.\ Bernardy for providing insight relevant to the completion of Section \ref{sec:multisymp}. We especially would like to thank F.\ Bonechi for suggesting that we consider the bracket $\{\cdot,\cdot\}^B$ defined in Proposition \ref{prop:bracket}, as it helped improve the content of our paper.
After completion we were told that I.\ Khavkine has a weaker version of Theorem \ref{Thm: Linfty}, point 2, which in particular states that $(C^\bullet_{\ham},\{\cdot,\cdot\}^B, D-\LQ)$ is a dgLa (i.e.\ without higher brackets). We could not find a published version of that statement, other than the notes of a talk by him.\footnote{It can be found here \href{https://users.math.cas.cz/~khavkine/talk-vienna-homotopy.pdf}{https://users.math.cas.cz/~khavkine/talk-vienna-homotopy.pdf}.} However, this is a simple consequence of Lemma \ref{Lem: PkillsJaco}, while our result mainly states that all $L_\infty$ algebras we construct are quasi-isomorphic (with explicit quasi-inverses).

\section{Preliminaries}

\subsection{Basics on local forms}
We fix a manifold $M$, equipped with a (possibly graded, see Section \ref{sec:BVBFV}) affine bundle $E\to M$. We denote the space of smooth sections of $E$ with $\mathcal{E}\doteq \Gamma(M,E)$.

The jet evaluation map is a smooth function into the infinite jet bundle \cite{Takens77,Zuckerman,Anderson,Delgado,BlohmannLFT}
\[
j^\infty\colon \mathcal{E}\times M \to J^\infty E.
\]
The complex of differential forms on $J^\infty E$ (which is thought as a pro-smooth manifold  \cite{Delgado}, a Diffeology \cite{BlohmannLFT}, or a Fr\'echet manifold \cite{krieglmichor}) splits into a bicomplex $\left(\Omega^{p,q}(J^\infty E),d_V^\infty,d_H^\infty\right)$, often called the \emph{variational bicomplex}, with respect to the vertical and horizontal differentials and, assuming $j^\infty$ is surjective, we can embed it (as a bicomplex) inside $\Omega^\bullet(\mathcal{E}\times M)$. 
\begin{definition}\label{def:localforms}
    The bicomplex of local forms is
    \[
        \left(\oloc^{\bullet,\bullet}(\mathcal{E}\times M)\doteq (j^\infty)^*\Omega^{\bullet,\bullet}(J^\infty E),  d_V  , d_H\right),
    \]
    where $ d_V  (j^\infty)^* \alpha = (j^\infty)^* d^\infty_V\alpha$ and $d_H (j^\infty)^* \alpha = (j^\infty)^* d^\infty_H\alpha$, with $d^\infty_V,d^\infty_H$ the vertical and horizontal differential in the variational bicomplex over $J^\infty E$. We will henceforth shorten $\oloc^{\bullet,\bullet}(\mathcal{E}\times M)\equiv \oloc^{\bullet,\bullet}$, and always understand $\mathcal{E}\times M$ when non-ambiguous.

    We also define:
    \begin{enumerate}
        \item The horizontal complex is $\Omega_H^\bullet \equiv \oloc^{0,\bullet}$ endowed with the horizontal differential $d_H$.
        \item The complexes of horizontal $p$-forms, for $p\geq 1$ are $(\oloc^{p,\bullet},d_H)$. 
        \item The higher horizontal complex is the direct sum over $p\geq 1$ of all complexes of higher horizontal $p$-forms, denoted $\left(\oloc^{\geq 1,\bullet},d_H\right)$.
    \end{enumerate}

    Given a vertically-homogeneous local form, of vertical form degree $p$, but inhomogeneous in the horizontal direction, we will denote it by $\bul{\alpha}\in\oloc^{p,\bullet}$. Furthermore, we denote by $\alpha^{k}$ its component of homogeneous \emph{co-form degree} $k$, i.e.\ $\alpha^k \in \oloc^{p,\top-k}$. 

    A local form $\bul{\alpha}\in\oloc^{p,\bullet}$ is called \emph{ultralocal} if it descends to a $(p,\bullet)$-form on $E$.

    A vector field $\bb{}\colon \mathcal{E}\to T\mathcal{E}$ is called local if it descends to a smooth bundle map covering the identity $X\colon J^kE \to VE$ for some $k$, where $VF\to M$ is the vertical bundle.
    \end{definition}    

    \begin{remark}
    A local vector field $\bb{}$ is often called \emph{evolutionary}, which is a terminology often reserved for the associated map $X$. Note that $X$ is not a vector field on $J^kE$ (see \cite[Section 5.1.4]{BlohmannLFT}). Importantly, local vector fields $X\in\Xloc(\mathcal{E})$ are such that $[\mathcal{L}_X,d_H] = 0$. Since $[d_H, d_V ]=0$ the previous condition is equivalent to $[\iota_X,d_H]=0$.
    \end{remark}

    \begin{theorem}[Anderson, Takens]\label{thm:sourceproj}
        There exists a unique operator $\Pi$ such that
        \begin{enumerate}
            \item $\bigoplus_{k=0}^{\top-1}\oloc^{p,k}\subseteq \ker \Pi$
            \item for all $\alpha\in \oloc^{p,\top}$, with $p\geq 1$, we have $\alpha - \Pi\alpha \in d\oloc^{p,\top-1}$,
            \item it is a projector: $\Pi^2 = \Pi$,
            \item it annihilates the image of the horizontal differential $\Pi \circ d_H=0$,
            \item it defines a differential $(\Pi \circ d_V)^2=0$.
        \end{enumerate}
        We call the image of $\Pi$ on $\oloc^{p,\bullet}$ the space of \emph{source} or \emph{functional} forms, and denote it by $\Omega_\src^{p,n}$.
    \end{theorem}

\subsection{Basics on homological algebra}
Let us collect here some standard arguments and concepts from homological algebra. 

\begin{definition}
    Let $(A,d_A)$ and $(B,d_B)$ be two cochain complexes and let $f\colon A\to B$  be a cochain map. 
    \begin{itemize}
        \item $f$ is called quasi-isomorphism, if it induces an isomorphism in cohomology. 
        \item The cochain complex $(C(f)=A[1]\oplus B,d_f)$ with 
            \begin{align*}
                d_f(a,b)=(-d_Aa,d_Bb+f(a))
            \end{align*}
        is called the cone of $f$. 
        \item $f$ is called homotopic to a cochain map 
            $\wt{f}\colon A\to B$, if there exists a 
            degree $-1$ map $h\colon A\to B$, such that
                \begin{align*}
                    f-\wt{f}=h\circ d_A +d_B \circ h. 
                \end{align*}
        \item $f$ is a homotopy equivalence, if there exists 
        $g\colon B\to A$ such that $g\circ f$ is homotopic to 
        $\id_A$ and $f\circ g$ is homotopic to 
        $\id_B$. 
        Pictorially, we denote homotopy equivalences by 
            \begin{equation}
          \xymatrix{
      h_A\circlearrowright (A, d_A) \ar@<1ex>[r]^-{f} & (B,d_B)\ar@<1ex>[l]^-g \circlearrowleft h_B.
       }
      \end{equation}
      A homotopy equivalence is called deformation retract, if either $h_A=0$ or $h_B=0$. A deformation retract is said to be \emph{special} iff, in addition
      \[
      h_A^2 = h_A\circ g = f \circ h_A = 0.
      \]
    \end{itemize}
\end{definition}

In this paper we are mainly interested in quasi-isomorphisms and, in particular, homotopy equivalences and constructions built upon these. The reason for that is that in field theory, and in particular within the BV-formalism, the actual \emph{physical} data are often given in terms of cohomology, which are then obviously preserved by quasi-isomorphisms. 

\begin{lemma}
Let     
\begin{equation}
          \xymatrix{
      h_A\circlearrowright (A, d_A) \ar@<1ex>[r]^-{f} & (B,d_B)\ar@<1ex>[l]^-g \circlearrowleft h_B.
       }
      \end{equation}
be a homotopy equivalence, then $f$ and $g$ are quasi-isomorphism which are inverse to each other in cohomology. Moreover, let 
    \begin{equation}
          \xymatrix{
      \wt{h}_B\circlearrowright (B, d_B) \ar@<1ex>[r]^-{i} & (C,d_C)\ar@<1ex>[l]^-j \circlearrowleft h_C.
       }
      \end{equation}
be another homotopy equivalence, then 
     \begin{equation}
          \xymatrix{
      H_A\circlearrowright (A, d_A) \ar@<1ex>[r]^-{i\circ f} & (C,d_C)\ar@<1ex>[l]^-{g\circ j} \circlearrowleft H_C.
       }
       \end{equation}
with $H_A=h_A- g\circ \wt{h}_B\circ f$ and $H_C=h_C-i\circ h_B\circ j$ is a homotopy equivalence as well. 
\end{lemma}
As a last general statement, we need the \emph{homological perturbation lemma}, which will prove useful for almost all of our construction, see \cite{crainic} for further details and references therein. 

\begin{theorem}[Homological pertubation lemma] 
\label{Thm: HomPtLem}
Let $f\colon A\to B$ and $g\colon B\to A$ be two cochain maps such that $f\circ g$ is homotopic to $\id_B$ via 
$h\colon B\to B$, i.e.\ 
    \begin{align*}
        \id_B-f\circ g= [d_B,h].
    \end{align*}
Moreover, let $k\colon B\to B$ be a degree $1$ map, such that 
$d_B+k$ is a differential. If $M=\id+kh$ is invertible, 
then we get for the maps 
    \begin{itemize}
    \item $\wt{f}=f+hM^{-1}kf$
    \item $\wt{g}=g+gM^{-1}kh$
    \item $\wt{h}=h+hM^{-1}kh$
    \item $\wt{d_A}=d_A + gM^{-1}kf$
    \end{itemize}
that $\wt{f}\colon (A,\wt{d}_A)\to (B,d_B+k)$ and $g\colon (B,d_B+k)\to (A,\wt{d}_A)$ are chain maps,
such that $\wt{f}\circ\wt{g}$ is homotopic to $\id_B$ via $\wt{h}$.
\end{theorem}

\begin{remark}
Throughout this paper, if we are in the situation 
of Theorem \ref{Thm: HomPtLem}, it is always the case that 
$\id+kh$ is invertible, because $kh$ is nilpotent, this means that the following sum is always finite 
    \begin{align*}
        (\id+kh)^{-1}=\sum_{i=0}^\infty (-kh)^i
    \end{align*}
and the corresponding perturbed maps are given by 
    \begin{itemize}
    \item $\wt{f}=(\sum_{i=0}^\infty (-hk)^i)f$,
    \item $\wt{g}=g(\sum_{i=0}^\infty (-kh)^i)$,
    \item $\wt{h}=h(\sum_{i=0}^\infty (-kh)^i)$ and 
    \item $\wt{d_A}=d_A + g(\sum_{i=0}^\infty (-kh)^i)kf$. 
    \end{itemize}
\end{remark}

\subsection{Basics on $L_\infty$ algebras}
$L_\infty$-algebras (also known as strong-homotopy Lie algebras) were first introduced in \cite{lada.stasheff:1993a,lada.markl:1994a} and play a prominent role in deformation theory.
In this section we recall the notions of $L_\infty$-algebras,
$L_\infty$-morphisms and their twists by Maurer--Cartan elements. 
For proofs and details we refer the reader to
\cite{kraft:2024}.

\begin{definition}($L_\infty$-algebra)\label{def:Linfty}
Let $\mathfrak{L}^\bullet$ be a graded vector space over $\mathbb{K}$. An \emph{$L_\infty$-structure} 
on $\mathfrak{L}^\bullet$ is a degree $+1$ coderivation $Q$ on the conilpotent
cocommutative coalgebra $S(\mathfrak{L}[1]^\bullet)$ cofreely cogenerated by $\mathfrak{L}[1]^\bullet$ such that $Q^2=0$. We call the pair $(\mathfrak{L}^\bullet,Q)$ an $L_\infty$-algebra.
\end{definition}

Here 
$S(\mathfrak{L}[1]^\bullet)$ denotes the  conilpotent cocommutative coalgebra cogenerated by a (graded) vector space 
$\mathfrak{L}[1]^\bullet$, which 
can be realised as the symmetrised deconcatenation
coproduct on the space $\bigoplus_{n\geq1}\bigvee^n\mathfrak{L}[1]^\bullet$
where $\bigvee^n\mathfrak{L}[1]^\bullet$ is the space of coinvariants for the
usual (graded) action of the symmetric group in $n$ letters $S_n$ on
$\otimes^n\mathfrak{L}^\bullet[1]^\bullet$, see e.g.\ \cite{kraft:2024}. Let us denote by 
\begin{equation}
	Q_n^k\colon \bigvee^n\mathfrak{L}[1]^\bullet\longrightarrow \bigvee^k\mathfrak{L}[1]^\bullet
\end{equation}
the components of the coderivation. 
 Any degree $+1$
coderivation $Q$ on $S(\mathfrak{L}[1]^\bullet)$ is then 
uniquely determined by the
components $Q^1_n$
through the formula 
\begin{equation}
	Q(\gamma_1\vee\ldots\vee\gamma_n)
	=
	\sum_{k=1}^n\sum_{\sigma\in\mbox{\tiny Sh($k$,$n-k$)}}
	\epsilon(\sigma)Q_k^1(\gamma_{\sigma(1)}\vee\ldots\vee
	\gamma_{\sigma(k)})\vee\gamma_{\sigma(k+1)}\vee
	\ldots\vee\gamma_{\sigma(n)}.
\end{equation} 
Here Sh($k$,$n-k$) denotes the set of $(k, n-k)$ shuffles in $S_n$, and
$\epsilon(\sigma)=\epsilon(\sigma,\gamma_1,\ldots,\gamma_n)$ is a sign
given by the rule $
\gamma_{\sigma(1)}\vee\ldots\vee\gamma_{\sigma(n)}=
\epsilon(\sigma)\gamma_1\vee\ldots\vee\gamma_n $.

The components $Q_n^1$ are often called higher brackets, and the condition $Q^2=0$ is equivalent to the higher Jacobi identities for the brackets
    \begin{equation}
	\label{eq:QsquaredZero}
	  \sum_{k=1}^n  \sum_{\sigma\in Sh(k,n-k)} 
		\epsilon(\sigma) 
		Q_{n-k+1}^1(Q_k^1(x_{\sigma(1)}\vee\cdots\vee x_{\sigma(k)})\vee 
		x_{\sigma(k+1)}\vee \dots \vee 	x_{\sigma(n)})
		=
		0.
	\end{equation}

    \begin{definition}
        Let $(\mathfrak{L}^\bullet,Q)$ be an $L_\infty$-algebra. An element $\pi\in \mathfrak{L}^1$ is called a \emph{Maurer--Cartan element} if it satisfies the equation
    	\begin{align*}
    	\sum_{k\geq 1} \frac{1}{k!}Q^1_k(\pi^{\vee k})=0.
    	\end{align*}
        We denote by $\mathsf{MC}(\mathfrak{L}^\bullet,Q)$ the set of Maurer--Cartan elements of the $L_\infty$-algebra $(\mathfrak{L}^\bullet,Q)$.
    \end{definition}

\begin{definition}[$L_\infty$-morphism]
Let us consider two $L_\infty$-algebras $(\mathfrak{L}^\bullet,Q)$ and
$(\wt{\mathfrak{L}}^\bullet,\wt{Q})$.  A degree $0$ 
coalgebra morphism
\begin{equation*}
	\Phi\colon 
	S(\mathfrak{L}[1]^\bullet)
	\longrightarrow 
	S(\wt{\mathfrak{L}}[1]^\bullet)
\end{equation*}
such that $\Phi Q = \wt{Q}\Phi$ is said to be a
$L_\infty$-morphism.
\end{definition} 

A coalgebra morphism $F$ from $S(\mathfrak{L}^\bullet)$ to
$S(\wt{\mathfrak{L}}^\bullet)$ is uniquely determined by its
components (also called \emph{Taylor coefficients})
\begin{equation*}
	\Phi_n^1\colon \bigvee^n\mathfrak{L}[1]^\bullet\longrightarrow \wt{\mathfrak{L}}[1]^\bullet,
\end{equation*}
where $n\geq 1$. Namely, we use the formula
\begin{equation*}
	\Phi(\gamma_1\vee\ldots\vee\gamma_n)=
\end{equation*}
\begin{equation}
	\label{coalgebramorphism}
	\sum_{p\geq1}\sum_{\substack{k_1,\ldots, k_p\geq1\\k_1+\ldots+k_p=n}}
	\sum_{\sigma\in \mbox{\tiny Sh($k_1$,..., $k_p$)}}\frac{\epsilon(\sigma)}{p!}
	\Phi_{k_1}^1(\gamma_{\sigma(1)}\vee\ldots\vee\gamma_{\sigma(k_1)})\vee\ldots\vee 
	\Phi_{k_p}^1(\gamma_{\sigma(n-k_p+1)}\vee\ldots\vee\gamma_{\sigma(n)}),
\end{equation}
where $\mathrm{Sh}(k_1,\dots,k_p)$ denotes the set of $(k_1,\ldots,
k_p)$-shuffles in $S_n$. For later use, we also denote 
    \begin{equation*}
	\Phi_n^k\colon \bigvee^n\mathfrak{L}[1]^\bullet\longrightarrow 
    \bigvee^k \wt{\mathfrak{L}}[1]^\bullet,
\end{equation*}
and therefore the condition of $\Phi$ being a morphism of $L_\infty$-algebras can be written as 
    \begin{align}
        \Phi^1_iQ^i_k=\tilde{Q}^1_i\Phi_k^i
    \end{align}
where we used Einstein's summation convention. Taking a Maurer-Cartan element 
$\pi\in \mathfrak{L}^1=\mathfrak{L}[1]^0$, we can use an $L_\infty$-morphism 
to obtain 
    \begin{align*}
        \Phi_\mathsf{MC}(\pi):=\sum_{k\geq 1}\frac{1}{k!} \Phi_k^1(\pi^{\vee k})
    \end{align*}
to obtain a Maurer-Cartan element in $\tilde{\mathfrak{L}}[1]$ (if we can make sense of the possibly infinite sum). 
Moreover, for a Maurer--Cartan element $\pi$ the structure maps 
    \begin{align*}
        Q^{1,\pi}_k(x_1\vee\dots\vee x_k):=
        \sum_{n=0}^\infty\frac{1}{n!} Q^{1}_{k+n}(\pi^{\vee n}\vee x_1\vee\dots\vee x_k)
    \end{align*}
define a new $L_\infty$-algebra structure on $\mathfrak{L}^\bullet$, called the $L_\infty$-algebra structure twisted by $\pi$. 
And we will denote the new $L_\infty$-algebra simply by $\mathfrak{L}^\pi$. 
If we have an $L_\infty$-morphism $\Phi\colon \mathfrak{L}^\bullet\to \tilde{\mathfrak{L}}^\bullet$ and a Maurer--Cartan element $\pi$, then the structure maps 
    \begin{align}\label{e:twistedmorphism}
        \Phi^\pi\colon  \mathfrak{L}^\pi\to \tilde{\mathfrak{L}}^{\Phi_{\mathsf{MC}}(\pi)},\qquad \Phi_k^{1,\pi}(x_1\vee\dots\vee x_k):=
        \sum_{n=0}^\infty\frac{1}{n!} \Phi^{1}_{k+n}(\pi^{\vee n}\vee x_1\vee\dots\vee x_k)
    \end{align}
define a morphism of $L_\infty$-algebras.

\begin{remark}
A dgL[k]a-structure (resp. $L_\infty[k]$-algebra structure) on a graded vector space $\mathfrak{L}^\bullet$ is the structure of a differential graded Lie algebra (resp. $L_\infty$-algebra) on $\mathfrak{L}^\bullet[k]$. 
If we have a $L_{\infty}[k]$-algebra $V$, then we have multilinear skew-symmetric maps $l_n\colon \bigwedge^n \mathfrak{L}^\bullet[k]\to \mathfrak{L}^\bullet[k]$ fulfilling a Jacobi identity, or equivalently by the maps 
    \begin{align*}
        \delta_n\colon S^n \mathfrak{L}^\bullet[k+1]\to \mathfrak{L}^\bullet[k+1]
    \end{align*}
given by the pre-concatenation with the décalage isomorphism which are the Taylor coefficients of the codifferential (cf.\ Definition \ref{def:Linfty}). 
\end{remark}

\section{Homotopies for local Lagrangian field theory}\label{sec:homVBC}
We lay here the groundwork that will lead to our formulation of local Lagrangian field theory. Our aim is to use the rich structure of the variational bicomplex to embed Lagrangian field theory in a consistent framework that allows us to work up to homotopy. 

We describe the space of equivalence classes of \emph{local Lagrangian densities} in terms of natural maps within the variational bicomplex, and build a resolution of it as a chain complex. In Section \ref{sec:Linfty} we will assign higher algebraic structures to the resolution built here, which will play an important role in field theory.

\subsection{Homotopies for local forms}
In \cite{BFLS} it was shown that, when $M$ is contractible and $E\to M$ is a vector bundle, the cone complex $\mathbb{R}[1]\oplus\oloc^{0,\bullet}$ for the canonical inclusion of constants into local forms is a resolution of the space of local functionals, defined as an appropriate space of equivalence classes of $(0,\top)$ local forms. In the mentioned literature, $(0,\top)$ forms are considered equivalent if they coincide on all sections of compact support (cf.\ Remark \ref{Rem: Stasheff}). In this section we review and extend the cited result, and we will give an equivalent, but sharper, definition of local functionals in terms of natural maps that arise from the cohomological analysis of the variational bicomplex (see Definition \ref{def:localfunctionals}).
Throughout the paper, if not stated differently, we always work with the \emph{total degree}, i.e.\ the sum of vertical, horizontal and internal degree (see also Section \ref{Sec:gradings} for further information on gradings).

We begin by stating a result by Anderson:
\begin{theorem}[\cite{Anderson}]\label{thm:sourcedecomposition}
    The higher horizontal complex $(\oloc^{\geq 1,\bullet},d_H)$ is a deformation retract of the trivial complex of $(\geq 1,\top)$ source forms:
\begin{equation}
\xymatrix{
      h^\nabla \circlearrowright (\oloc^{\geq 1,\bullet},d_H) \ar@<1ex>[r]^-{\Pi} & \Omega_\src^{\geq 1, n} \ar@<1ex>[l]^-I 
    }, \qquad [d_H,h^\nabla]+I\circ \Pi=\id
\end{equation}
where $h^\nabla$ is Anderson's global homotopy (depending on a symmetric connection $\nabla$), the map $\Pi$ denotes the projector onto source forms of Theorem \ref{thm:sourceproj}, and $I$ is the trivial inclusion. Moreover:
    \begin{align*}
        \Pi\circ I=\id, \qquad \Pi\circ h^\nabla=0. 
    \end{align*}
\end{theorem}
\begin{proof}
    Anderson proves an analogous theorem in the variational bicomplex $\Omega^{\bullet,\bullet}(J^\infty E)$. Since all maps are local in his construction, we can simply pullback along $j^\infty$ and obtain the desired result.
\end{proof}

Now, using the homological perturbation lemma (see \cite{crainic}), 
we can perturb $d_H$ in  the direction $d_V$ to obtain the de Rham differential (the total differential on the variational bicomplex), 
\begin{proposition}
    There is a deformation retract
\begin{equation}\label{e:perturbeddefretract}
\xymatrix{
      \wt{h}^\nabla \circlearrowright (\oloc^{\geq 1,\bullet},d) \ar@<1ex>[r]^-{\wt{\Pi}_\src} & (\Omega_\src^{\geq 1, n},\wt{d}_V) \ar@<1ex>[l]^-{\wt{I}} 
    }
\end{equation} where 
    \begin{itemize}
        \item $\wt{\Pi}_\src=\Pi\sum_{k=0}^{\dim M}(-d_Vh^\nabla)^k= \Pi$
        \item $\wt{I}=(\sum_{k=0}^{\dim M} (-h^\nabla d_V)^k)I$
        \item $\wt{h}^\nabla=h^\nabla\sum_{k=0}^{\dim M} (-d_Vh^\nabla)^k $
        \item $\wt{d}_V=\Pi(\sum_{k=0}^{\dim M} (-d_Vh^\nabla)^k) d_V I=\Pi d_V I$.
    \end{itemize}
\end{proposition}
\begin{proof}
    This is a standard application of the homological perturbation lemma, after the observation that the map $d_V h^\nabla\colon \oloc^{s,r} \mapsto\oloc^{s+1,r-1}$ is nilpotent, since $\oloc^{s,-i}=0$ for all $i\geq 1$. Using the formulas from the homological pertubation lemma, we get the sought deformation retract.
    
    Moreover, in the first and the last point of the statement we used that $\Pi$ vanishes if the horizontal degree is not $\dim M$. Moreover, this implies also that $\Pi\circ\wt{I}=\id$, which does not automatically follow from the homological perturbation lemma. 
\end{proof}

\begin{remark}
    If we consider the vertical differential as a cochain morphism 
    \[
    d_V\colon (\ohor^\bullet[-1],-d_H)\to (\oloc^{\geq 1,\bullet},d)
    \]
    then its cone $C(d_V)$ is canonically isomorphic to 
    $\oloc^{\bullet,\bullet}$, which is quasisomorphic to $(\Omega^\bullet(E),d)$. On the other hand, by the discussion before, the cone of $d_V$ as described above is quasi-isomorphic to the cone of $\Pi\circ d_V\colon (\ohor^\bullet[-1],-d_H)\to (\Omega_\src^{\geq 1,n},\wt{d}_V)$, since $\Pi\colon \oloc^{\geq 1,\bullet} \to \Omega_{\src}^{\geq 1,\bullet}$ is a quasisomorphism in virtue of Theorem \ref{thm:sourcedecomposition}. We obtain the cone $C(\Pi\circ d_V)\doteq\ohor^\bullet\oplus \Omega_\src^{\geq 1,n}$, endowed with the differential 
        \begin{align*}
            d_C(\alpha,\beta)=(d_H\alpha,\wt{d_V}\beta+\Pi d_V \alpha)
        \end{align*}
    Hence, comparing degrees of $\ohor^\bullet$ and $\Omega_\src^{\geq 1, n}$, one can see that 
        \begin{equation}
            \begin{cases}
                \mathrm{H}^i_\mathrm{H}=\mathrm{H}^i(E) & \text{for } i\in \{0,\dots,\dim M-1\}\\
                \frac{\ker \Pi\circ d_V\colon \ohor^{n}\to \Omega_\src^{1,n}}{\im d_H\colon \ohor^{n-1}\to \ohor^{n}} =\mathrm{H}^n(E) & \text{for } i=n \\
                \mathrm{H}_\src^j=\mathrm{H}^j(E)& \text{for } j>n
            \end{cases}
        \end{equation}
    This proves a statement appearing in (the unpublished notes) \cite[Theorem 5.2.6]{BlohmannLFT}, which also clarifies a statement appearing in \cite[Theorem 5.9]{Anderson}. In the literature, the operator $\Pi$ is often called the interior Euler operator, and $\Pi\circ d_V$ is called the exterior Euler operator.
\end{remark}

Let us now assume that $E\to M$ and thus also $J^\infty E$ is a vector bundle over $M$.
\begin{theorem}[Anderson]\label{thm:verticalhomotopy}
    There exists a special deformation retract
    \[
    \xymatrix{
      h_V \circlearrowright (\oloc^{\bullet,\bullet},d_V) \ar@<1ex>[r]^-{0^*} & (\Omega^{\bullet}(M),d_M) \ar@<1ex>[l]^-{p^*} 
    }
    \]
    where we denoted the $p\colon J^\infty E \to M$ the vector bundle projection and by $0\colon M\to J^\infty E$ the zero-section. 
    
    Moreover, the \emph{vertical homotopy} $h_V\colon \oloc^{\bullet,\bullet}\to \oloc^{\bullet-1,\bullet}$ commutes with the horizontal differential, so that 
    \begin{align*}
        [d,h_V]+p^*0^*=[d_V,h_V]+p^*0^*=\id
    \end{align*}
    and the special deformation retract conditions 
    \[
    h_V^2=h_Vp^* = 0^*h_V = 0
    \]
    are satisfied.
\end{theorem}

\begin{remark}[Extension to affine bundles]
    Theorem \ref{thm:verticalhomotopy} and the rest of our paper can be generalised to the case of affine bundles, by replacing the zero section $0\colon M \to J^\infty E$ with the jet prolongation $j^\infty s_0$ of any reference section $s_0\colon M \to E$. This is particularly relevant for it allows us to generalise the discussion to standard \emph{gauge} theories, usually formulated in terms of connections on a principal bundle $P\to M$, since they can also be seen as sections of the affine bundle $E=TP/G$ \cite{Kobayaschi,MangiarottiSardanashvily2000}. Throughout, we shall keep the notation introduced above for simplicity.
\end{remark}

Let us now introduce a new complex:
\begin{definition}
    The horizontal cone is $C^\bullet(p^*)\doteq \Omega(M)^\bullet[1]\oplus  \ohor^\bullet$ with differential
    \begin{align*}
       D\colon C^\bullet(p^*)\to C^{\bullet+1}(p^*),\qquad D(\alpha,\omega)=(d\alpha,-d_H\omega+p^*\alpha).
    \end{align*}
    Denote by $\pi_2\colon C^\bullet(p^*)\to \ohor^\bullet$.
\end{definition}

\begin{theorem}
    There is a homotopy equivalence
    \[
    \xymatrix{
    H_{0^*} \circlearrowright(\Omega(M)^\bullet[1]\oplus  \ohor^\bullet,D)  \ar@<1ex>[r]^-{I_V}    
        &
    (\oloc^{\geq 1,\bullet}[1],d) \circlearrowleft H_V \ar@<1ex>[l]^-{P_V} 
    }
    \]
    where $I_V=\pi_2(0\oplus d_V)$ and $P_V=i_2\circ h_V\circ p^{1,\bullet}$, with $p^{1,\bullet}$ the projection to forms of vertical degree exactly $1$, while 
    \[
    H_{0^*}\colon (\Omega(M)[1]\oplus  \ohor)^\bullet \to  (\Omega(M)[1]\oplus  \ohor)^{\bullet-1}, \quad (\alpha,\omega)\mapsto (0^*\omega,0)
    \]
    and
    \[
    H_V\colon \oloc^{p\geq 1,\bullet}\to \oloc^{p\geq 1,\bullet}, \qquad 
    H_V=\begin{cases}
        h_V & p>1\\
        0 & \text{else}
    \end{cases}.
    \]
    In particular, then
    \begin{align*}
        [d,H_V]+ I_V\circ P_V =\id, \qquad
        [D,H_{0^*}] + P_V\circ I_V =\id.
    \end{align*}
\end{theorem}

\begin{proof}
This is a consequence of the general construction.\end{proof}

This tells us that, using general results
\begin{corollary}
    There is a diagram of homotopy equivalences 
    \[
    \xymatrix{
    (\Omega(M)^\bullet[1]\oplus  \ohor^\bullet,D) \ar@<1ex>[r]^-{I_V}    
        &
    (\oloc^{\geq 1,\bullet}[1],d)  \ar@<1ex>[l]^-{P_V} \ar@<1ex>[r]^-{\Pi} 
            & (\Omega_\src^{\geq 1, n}[1],\wt{d}_V) \ar@<1ex>[l]^-{\wt{I}},
    }
    \]
    which yields a homotopy equivalence
    \[   
        \xymatrix{
          H \circlearrowright (\Omega(M)^\bullet[1]\oplus  \ohor^\bullet,D)\ar@<1ex>[r]^-{\Pi\circ I_V} & (\Omega_\src^{\geq 1, n}[1],\wt{d}_V) \ar@<1ex>[l]^-{P_V\circ\wt{I}} \circlearrowleft \wt{h}_V
        }.
    \]
    where the homotopies are given by 
    \begin{align*}
        H=H_{0^*} - P_V\circ \wt{h}^\nabla\circ I_V \colon (\Omega(M)[1]\oplus  \ohor)^\bullet \to (\Omega(M)[1]\oplus  \ohor)^{\bullet-1}
    \end{align*}
    and
    \begin{align*}
        \wt{h}_V=-\Pi \circ H_V \circ \wt{I}
        =-\Pi \circ H_V \circ I
        \colon \Omega_\src^{\geq 1, n}[1]^\bullet \to \Omega_\src^{\geq 1, n}[1]^{\bullet-1}.
    \end{align*}
\end{corollary}

\begin{remark}
Let us make some observations 
    \begin{itemize}
        \item The complex $(\Omega(M)^\bullet[1]\oplus  \ohor^\bullet,D)$ is concentrated in degrees $\{-1,\dots,\dim M\}$ and the complex $\Omega_\src^{\geq 1, n}[1]$ is concentrated in degree bigger than or equal to $\dim M$, so the cohomologies of both are concentrated in degree $\dim M$. 
        \item Using the previous point and the equation for the homotopy equivalence, we get that 
        \[
        \id= (\Pi\circ I_V)\circ (P_V\circ \wt{I})\colon \Omega_\src^{1, n}[1]^{cl}\to \Omega_\src^{1, n}[1]^{cl},
        \]
        and conclude that  $ (\Pi\circ I_V)$ is surjective and $P_V\circ \wt{I}$ is injective. 
        \item The map $P_V$ concatenates $p^{1,\bullet}$, whose kernel is any vertical form of degree higher than $1$, with $h_V$. Hence, we have that 
        \[
        P_V \circ \wt{I} = i_2 \circ h_V \circ p^{1,\bullet} \circ \wt{I} = i_2 \circ h_V \circ I.
        \]
    \end{itemize}
\end{remark}

Let us define the space of functional $(1,n)$ forms as
\[
\mathcal{F}^n\doteq \ker(\Pi \circ d_V\colon \Omega^{1,n}_\src\to \Omega^{2,n}_\src).
\]
From the discussion before, it is clear that we can restrict our 
homotopy equivalence to a deformation retract 
    \[   
    \xymatrix{
      H \circlearrowright (\Omega(M)^\bullet[1]\oplus  \ohor^\bullet,D)\ar@<1ex>[r]^-{\Pi\circ I_V} & \mathcal{F}^n \ar@<1ex>[l]^-{P_V\circ\wt{I}}, 
    }.
\]
where we see $\mathcal{F}^n$ as a cochain complex with trivial differential. 
Consider the space
\[
\im(P_V\circ \wt{I}) = \im(i_2\circ h_V \circ I)\subset \Omega^\bullet(M)[1]\oplus\ohor^{n} \subset C^\bullet(p^*).
\]
Because of the discussion above, the map 
    \begin{align*}
        \mathbb{p} = (P_V\circ \wt{I})\circ(\Pi\circ I_V) \equiv (i_2\circ h_V\circ I)\circ(\Pi\circ I_V)
    \end{align*}
is a projector with image $\im(P_V\circ \wt{I})\equiv 0\oplus \im(h_V\circ I)$, which is canonically isomorphic to $\mathcal{F}^n$ and therefore to the cohomology of $\Omega_\src^{\geq 1, n}$. Note, furthermore, that we canonically also have 
\begin{align}
    [D,H]+\mathbb{p} =\id,
\end{align}
whence we obtain the deformation retract 
    \begin{equation}
        \xymatrix{
          H \circlearrowright (\Omega(M)^\bullet[1]\oplus  \ohor^\bullet, D)\ar@<1ex>[r]^-{\mathbb{p}} & \im(P_V\circ \wt{I}) \ar@<1ex>[l]^-{i_{\mathrm{can}}}
        },
    \end{equation}
    with $i_{\mathrm{can}}$ the canonical inclusion of a subspace. This suggests the following:
    
    \begin{definition}\label{def:localfunctionals}
    Denote by $\bbP=\pi_2\circ \mathbb{p}=h_V \circ I\circ \Pi \circ I_V$.
    The space of local functionals is $\calF \doteq \im (\bbP) = \im (h_V \circ I\circ \Pi \circ I_V)\equiv \pi_2(\im(P_V\circ \wt{I}))\subset \ohor^n$.
    \end{definition}

    Collecting the results above we come to the main result in this section.
\begin{theorem}[Resolution of local functionals]\label{thm:resolution}
    The horizontal cone is a deformation retract of the space of local functionals $\mathcal{F}\doteq\im(\bbP)$
    \begin{equation}\label{e:locformdefret}
        \xymatrix{
          H \circlearrowright (\Omega(M)^\bullet[1]\oplus  \ohor^\bullet, D)\ar@<1ex>[r]^-{\bbP} & \calF \ar@<1ex>[l]^-{i}
        },
    \end{equation}
    where $i\colon \mathcal{F}\hookrightarrow \Omega(M)^\bullet[1]\oplus  \ohor^\bullet$ is the inclusion of $\im(\bbP)\subset \ohor^n\subset \Omega(M)^\bullet[1]\oplus  \ohor^\bullet$.
\end{theorem}

\begin{remark}[Local functionals: comparison of definitions]\label{Rem: Stasheff}
    Local functionals are usually defined as equivalence classes of $(0,\top)$-forms that vanish along the zero section,\footnote{Note that this means discarding constant Lagrangians, which do not give rise to (nontrivial) equations of motion.} i.e.\ which are in the kernel of $0^*$, and the equivalence relation is given by
    \begin{align*}
        (j^\infty)^*\beta\sim (j^\infty)^*\beta' \colon\iff \int_M (j^\infty\phi)^*(\beta-\beta')=0
    \end{align*}
for all compactly supported $\phi \in \mathcal{E}$. It is clear that 
the image of $\bbP $ is contained in the kernel of $0^*$, since $h_V$ is the homotopy of a special deformation retract (Theorem \ref{thm:verticalhomotopy}) and $0^*\circ h_V=0$. Let us now assume that we have a $\beta\in \ohor^n$, such that $0^*\beta=0$ and $\beta\sim 0$, which implies that $\beta=d_H\alpha$ (see \cite[Lemma 10]{BFLS} who cite \cite{Olver}) or, equivalently, $D(0,\alpha)=(0,\beta)$ and thus 
$\beta\in \ker \bbP $. On the other hand let 
$\beta\in \ker \bbP $ and $0^*\beta=0$, then by the explicit form of the homotopy, there exists an $\alpha\in \ohor^{n-1}$, such that
    \begin{align*}
        D(0,\alpha)=(0,d_H\alpha)=(0,\beta),\qquad 0^*\alpha=0,
    \end{align*}
and therefore we get 
    \begin{align*}
        \int_Mj^\infty\phi^*\beta=\int_Mj^\infty\phi^*d_H\alpha
        =\int_Md j^\infty\phi^*\alpha=0,
    \end{align*}
for all compactly supported sections $\phi\in \mathcal{E}_c$.
This implies that our definition coincides with the usual definition. Note that the condition that the differential forms vanish along the zero-section ensures that $\mathrm{supp} (j^\infty\phi^*\beta)\subseteq \mathrm{supp}(\phi)$ and the integral is well-defined for compactly supported sections.

In case that the manifold $M$ is contractible, 
we can find a deformation retract 
 \begin{equation}
        \xymatrix{
          h_M\circlearrowright (\Omega(M)^\bullet, D)\ar@<1ex>[r] & \mathbb{R}\ar@<1ex>[l] 
        },
    \end{equation}
and with this we find that the horizontal cone is quasi-isomorphic to $\mathbb{R}[1]\oplus \ohor^\bullet$, which is exactly the resolution used in \cite{BFLS}. 
\end{remark}

Sometimes it can be also useful to \emph{remove the constants}, i.e.\ we consider the $\ker 0^*\subseteq \ohor^\bullet$ as a subspace of $\Omega(M)[1]\oplus \ohor$
by including it in the second summand.
\begin{theorem}
\label{Thm: DefRetractwithoutconstants}
    The subspace  $0\oplus\ker 0^*\subseteq \Omega(M)[1]\oplus\ohor^\bullet$ is a subcomplex, with $H(0\oplus\ker 0^*)\subseteq 0\oplus\ker 0^*$ and $\im(\mathbb{p})\subseteq 0\oplus\ker 0^*$ and the induced diagram 
        \begin{align*}
            \xymatrix{
          H \circlearrowright (0\oplus\ker 0^*, D)\ar@<1ex>[r]^-{\bbP } & \mathcal{F} \ar@<1ex>[l]^-{i}
        }
        \end{align*}
is a deformation retract, with $i$ the canonical inclusion, as above. Hence, there is a deformation retract
\[
\xymatrix{
          h \circlearrowright (\ker 0^*, d_H)\ar@<1ex>[r]^-{\bbP_0} & \mathcal{F} \ar@<1ex>[l]^-{i}.
        }
\]
\end{theorem}

\begin{proof}
    Let us start with the differential
    \begin{align*}
        D(0,\bL)=(0,d_H\bL),
    \end{align*}
so if $\bL\in \ker 0^*$, also $d_H\bL\in \ker 0^*$, since $0^*$ is a chain map. For the homotopy we get 
    \begin{align*}
        H(0,\bL)=(0,-h_Vh^\nabla d_V\bL),
    \end{align*}
and since $0^*h_V=0$, we get that $H(\ker 0^*)\subseteq \ker 0^*$. Moreover, since 
\[
\mathbb{p} (\alpha,\bL)=(0, h_VI\Pi d_V\bL),
\]
the same argument holds. So the only thing we need to check is that 
    \begin{align*}
        \mathrm{Im}(\bbP |_{\ker 0^*})=\mathrm{Im}(\bbP ),
    \end{align*}
so let $(\alpha,\bL)\in \Omega(M)[1]\oplus\ohor$, then $(0,\bL-p^*0^*\bL)\in {\ker 0^*}$ and 
\[
\bbP (\alpha,\bL)= h_VI\Pi d_V \bL = h_VI\Pi d_V(\bL-p^*0^*\bL) = \bbP (0,\bL-p^*0^*\bL)
\] 
since $d_Vp^*0^*\bL=0$. The second statement follows by defining $\bbP\vert_{0\oplus\ker{0^*}}=0\oplus \bbP_0$ and $h \bL=-h_Vh^\nabla d_V \bL$.
\end{proof}

\begin{remark}
    In \cite{BFLS} the authors also give two variants of resolutions of the local functionals, for the case of $M$ contractible. In view of Remark \ref{Rem: Stasheff}, Theorem \ref{Thm: DefRetractwithoutconstants} provides the other resolution, but again now in the non-contractible case. 
\end{remark}

\subsection{Graded bundles and gradings}
\label{Sec:gradings}
The affine bundle $E\to M$ that is typically considered when working within the BV framework---our main application (see Section \ref{sec:BVBFV})---is $\mathbb{Z}$-graded,\footnote{We could allow the bundle to be nonlinear in degree zero. Since our results are mainly for vector and affine bundles, we will make this slightly restrictive assumption.} meaning that we consider an affine bundle in the category of graded manifolds, with base $M$ of degree $0$. Note that from the point of view of the homotopy structure of the variational bicomplex nothing changes in the graded case, so we can simply port the discussion outlined so far to the graded scenario, using the same symbols.

However, owing to this additional ``internal'' grading, the space $\oloc^{\bullet,\bullet}(\mathcal{E}\times M)$ comes equipped several relevant gradings and combinations thereof, so it is convenient to introduce the following disambiguation:
\begin{definition}[Gradings]
~
\begin{enumerate}
    \item An element $\alpha\in\oloc^{\bullet,\bullet}(\mathcal{E}\times M)$ has \emph{vertical form degree} $\mathsf{vfd}(\alpha)=p$ iff $\alpha\in\oloc^{p,\bullet}(\mathcal{E}\times M)$.
    \item An element $\alpha\in\oloc^{\bullet,\bullet}(\mathcal{E}\times M)$ has \emph{horizontal co-form degree} $\mathsf{hcd}(\alpha)=k$ iff $\alpha\in\oloc^{\bullet,\mathrm{top}-k}(\mathcal{E}\times M)$. 
    \item An element $\alpha\in\oloc^{\bullet,\bullet}(\mathcal{E}\times M)$ has \emph{total form degree} $\mathsf{tfd}(\alpha) = k$ iff $\alpha\in\oloc^{p, q}(\mathcal{E}\times M)$ with $k=p +q$ and $\mathsf{tfd} = \mathsf{vfd} - \mathsf{hcd} + \mathrm{dim}(M)$. We will call the \emph{effective form degree} of a homogeneous element $\alpha\in\oloc^{p,q}(\mathcal{E}\times M)$ the number $\mathsf{efd}\doteq \mathsf{vfd} - \mathsf{hcd}$, so that 
    \[
    \mathsf{tfd}=\mathsf{efd} + \mathrm{dim}(M).
    \]
    \item The internal degree specified by $F$ is called \emph{ghost degree}\footnote{Also known as ghost number, in the literature.}, and it is denoted by $\mathsf{ghd}(\alpha)$. 
    \item The \emph{partial effective degree} of an element $\alpha$, homogeneous in ghost degree and horizontal coform degree, is the difference $\mathsf{ped}=\mathsf{ghd} - \mathsf{hcd}$.
    \item The \emph{total effective degree} of an element is the sum of the partial degree and the vertical form degree $\mathsf{ted}= \mathsf{ped} + \mathsf{vfd}$, and we have the following identity: 
    \[
    \mathsf{ted}= \mathsf{ped} + \mathsf{vfd} = \mathsf{ghd} - \mathsf{hcd} + \mathsf{vfd} = \mathsf{ghd} + \mathsf{efd}. 
    \]
\end{enumerate}
\end{definition}

Note that, restricted to horizontal forms the total degree, which we elected to use preferentially throughout, differs from the partial effective degree by the dimension of the manifold. We will use the degree $\mathsf{ped}$ when talking about skew symmetry of graded symplectic forms, and as a useful bookkeeping device later on.

\begin{definition}[Developments of forms]\label{def:developments}
    Given a local $p$-form $\alpha\in\oloc^{p,\top}$, a \emph{development} of $\alpha$ is a horizontally inhomogenous form $\bul{\alpha}$ such that it is homogeneous of partial effective degree $\mathsf{ped}(\bul{\alpha})=\mathsf{ped}(\alpha^0)$, and $\alpha^0=\alpha$.
\end{definition}

The notion of form development, for purposes similar to those treated in this paper, was studied by \cite{Sharapov1,Sharapov2}. See also \cite{Grigoriev} and references therein. It relates, as we will see, to the notion of descent equations, see e.g.\ \cite{ZUMINO1985477,ZuminoDescent}.


\begin{definition}
The \emph{graded Euler vector field} $\mathsf{E}\in\Xloc(\mathcal{E})$ is the ghost degree-$0$ vector field that acts on ghost-number-homogeneous local forms (and vector fields) by 
\[
\LE \alpha = \mathsf{ghd}(\alpha) \alpha.
\]
\end{definition}

\begin{lemma}[{\cite[Lemma 18]{MSW}}]\label{QEuler}
Given $X\in\Xloc(\mathcal{E})$ homogeneous of ghost degree $\mathsf{ghd}(X)=1$, evolutionary, vector field, and $\mathsf{E}$ the graded Euler vector field, we have
\[
\mathcal{L}_X = \mathcal{L}_{[\mathsf{E},X]}=\LE\mathcal{L}_X - \mathcal{L}_X\LE, \qquad \mathcal{L}_X\LE = (\LE - 1) \mathcal{L}_X.
\]
\end{lemma}

\section{Strong-homotopy Lie algebras for field theory}\label{sec:Linfty}
In many applications of interest, one wants to endow $\calF$ with symplectic data. The usual way to do this, on \emph{closed} manifolds, is to identify $\calF$ with functions of the form $\int_M L$ with $L\in \oloc^{0,\top}$. The integration operation descends to the quotient $\oloc^{0,\top}/d\oloc^{0,\top-1}$, and one can look for a weak-symplectic structure of the form $\int_M \omega$ with $\omega\in\oloc^{2,\top}$. This induces a Lie algebra on the space of local functionals. Having constructed a resolution of $\calF$, we look for symplectic data thereon, and for the natural (possibly higher) Lie-structures associated to it.

This problem was treated in detail also in \cite{BFLS}, in the contractible case and using homological perturbation. Our preliminary constructions, and in particular the control on homotopies, will provide us with explicit formulas for higher brackets, and extend the results to noncontractible bases.

\subsection{Symplectic and Hamiltonian forms}
\begin{definition}
A local form $\omega \in\oloc^{2,\top}(\mathcal{E}\times M)$ of ghost degree $\mathsf{ghd}(\omega)=k$ is said to be a $k$-symplectic local form if it is $d_V$ (and $d_H$) closed, and if the  induced  map 
    \begin{align*}
        \Pi\circ(\omega)^\flat \colon \Xloc(\mathcal{E})\to 
        \Omega_\src^{1,n}
    \end{align*}
is injective.
A $k$-symplectic development of $\omega$ is a development $\bom$ in the sense of Definition \ref{def:developments} such that $d_V\bom=0$. 
\end{definition}

\begin{definition}\label{def:Hamiltonian}
    Let $\bom\in\oloc^{2,\bullet}$ a local $k$-symplectic development of $\omega$. The space of Hamiltonian $(0,\bullet)$ forms for $\bom$ is
    \begin{align*}
        \oham^\bullet=\{\bF\in \ohor^\bullet\mid \exists \bb{\bF}\in \mathfrak {X}(\mathcal{E}): \ \Pi\iota_{\bb{\bF}}\bom =\Pi d_V \bF\}.
\end{align*}
We call $(\bF,\bb{\bF})$ a Hamiltonian pair for $\bom$.
\end{definition}

We will see now that there is a particular type of inhomogeneous $k$-symplectic form, arising from the interaction between $\omega$ and a cohomological vector field $Q\in\Xloc(\mathcal{E})^1$. Before we proceed, however, we will need the following
\begin{lemma}\label{lem:nonhamiltooniandefretract}
    Let $Q\in\Xloc(\mathcal{E})^1$ cohomological. Then, there is a deformation retract
    \begin{equation}
    \xymatrix{
      \wt{h}^\nabla \circlearrowright (\oloc^{\geq 1,\bullet},d_H-\LQ) \ar@<1ex>[r]^-{\wt{\Pi}} & (\Omega_\src^{\geq 1, n},\wt{d}_Q) \ar@<1ex>[l]^-{\wt{I}} 
    }
    \end{equation}
    with 
    \begin{itemize}
        \item $ \wt{\Pi}_\src=\Pi\circ
        (\sum_{k\geq 0}(\LQ h^\nabla)^k)=\Pi$,
        \item $\wt{i}= (\sum_{k\geq 0}(h^\nabla\LQ)^k)\circ \iota$,
        \item $\wt{h}^\nabla=h^\nabla\sum_{k\geq 0}(\LQ h^\nabla)^k)$,
        \item $\wt{d}_Q=-\Pi \circ \LQ \circ I$. 
    \end{itemize}
\end{lemma}
\begin{proof}
    We simply need to perturb the datum
    \begin{equation}
    \xymatrix{
      h^\nabla \circlearrowright (\oloc^{\geq 1,\bullet},d_H) \ar@<1ex>[r]^-{\Pi} & \Omega_\src^{\geq 1, n} \ar@<1ex>[l]^-I 
    }
\end{equation}
by the homological perturbation lemma in the direction of $-\LQ$. Observe that $\wt{\Pi}_\src=\Pi$, since the latter vanishes on $\oloc^{\geq 1, <\top}$, a fact that is also used to show that $\wt{d}\equiv -\Pi\circ \sum(\LQ h^\nabla)^k \LQ I = -\Pi\circ \LQ \circ I$.
\end{proof}

\begin{theorem}\label{thm:extomega}
Let $\omega$ be a local symplectic form of ghost degree $\mathsf{ghd}(\omega)=k$. Moreover, assume $Q\in\Xloc(\mathcal{E})^1$ is an odd, cohomological vector field with $\mathsf{ghd}(Q)=1$ and such that 
    \begin{align*}
        \Pi(\LQ\omega)=0. 
    \end{align*}
Then there exists a $k$-symplectic development $\bom$ of $\omega$, such that 
\begin{equation}\label{e:specialsymplecticdevelopment}
(d_H - \LQ) \bom = 0. 
\end{equation}
Moreover, if $\omega^\bullet$ and $\wt{\omega}^\bullet$ are two $k$-symplectic developments of $\omega$, such that both satisfy \eqref{e:specialsymplecticdevelopment}, then there exists a 
$\eta^\bullet\in\oloc^{2,\leq \top-1}$ such that $\mathsf{ped}(\bul{\eta})=k-1$ and
    \begin{align}\label{e:differentdevelopments}
    \wt{\omega}^\bullet = \omega^\bullet  +(d_H-\LQ)\bul{\eta}.  
    \end{align}
Finally, if $\omega$ is ultralocal (Definition \ref{def:localforms}), there is an isomorphism of complexes
    \[
    \Pi\circ (\omega)^\flat\colon (\Xloc,-\LQ)\to (\Omega_{\src}^{1,n},\wt{d}_Q).
    \]
    and a quasi-isomorphism
    \[
    (\omega)^\flat\colon (\Xloc,-\LQ)\to (\oloc^{1,n},d_H-\LQ).
    \]
\end{theorem}

\begin{proof}
Let us construct a development $\bom$ of $\omega$, by 
    \begin{align*}
        \bom =\sum_{k}(d_Vh_Vh^\nabla\LQ)^k\omega.
    \end{align*}
It is clear that $d_V\bom=0$, since $d_V\omega=0$ and the higher terms are $d_V$-exact.  
Let us write $\bom=\sum_{k\geq 0}\omega^k$ where 
$\omega^k=(d_Vh_Vh^\nabla\LQ)^k\omega$. Counting horizontal coform degree, we get that $(d_H - \LQ) \bom = 0$ is equivalent to 
    \begin{align*}
        \LQ\omega^k=d_H\omega^{k+1}
    \end{align*}
for all $k$. Let us assume that this equation is fulfilled for all $k\leq n$, then 
\begin{align*}
    d_H\omega^{k+2}&=d_H (d_Vh_Vh^\nabla \LQ)\omega^{k+1}
    =
    d_Vh_Vd_Hh^\nabla \LQ \omega^{k+1}\\&
    =d_Vh_V\LQ\omega^{k+1}- 
     d_Vh_Vh^\nabla d_H \LQ \omega^{k+1}
     -d_Vh_V\Pi \LQ\omega^{k+1}\\&
     = d_Vh_V\LQ\omega^{k+1} - 
     d_Vh_Vh^\nabla \LQ d_H \omega^{k+1}\\&
     =d_Vh_V\LQ\omega^{k+1} - 
     d_Vh_Vh^\nabla \LQ \LQ\omega^{k}\\&
     =d_Vh_V\LQ\omega^{k+1} =\LQ\omega^{k+1}
     -h_Vd_V\LQ\omega^{k+1}\\&
     =\LQ\omega^{k+1}
\end{align*}
since for $k=0$ the equation is canonically fulfilled, observing that the partial effective degree is preserved by $(d_H - \LQ)$, we get the existence of such a $k$-symplectic development. 

Now, Lemma \ref{lem:nonhamiltooniandefretract} allows us to write
    \begin{align*}
        \bom -\bul{\wt{\omega}} &
        =
        ([d_H-\LQ,\wt{h}^\nabla]+\wt{i}\circ \Pi)(\bom -\bul{\wt{\omega}})\\&
        =
        (d_H-\LQ)\wt{h}^\nabla(\bom -\bul{\wt{\omega}}),
    \end{align*}
where we used that $\Pi$ of $\bom$ and $\bul{\wt{\omega}}$ coincide, since they both extend $\omega$ and that both are $(d_H-\LQ)$-closed. Note that the homotopy $h^\nabla$ decreases the horizontal degree by one and thus the claim of Equation \ref{e:differentdevelopments} is proven. 

Finally, if $\omega$ is a symplectic form, the map $\Pi\circ(\omega)^\flat\colon \Xloc\to\Omega_{\src}^{1,n}$ is injective. If $\omega$ is ultralocal it means it is of the form
\[
\omega = (j^\infty)^*\left(\sum_{i=1}^{\mathrm{rk}(E)}\omega_{ij} d_Vu^i\wedge d_Vu^j\right)
\]
where $\omega_{ij}\colon J^\infty E \to GL(\mathrm{rk}(E))$ is a smooth map, and $\{u^i\}$ denotes a coordinate chart in $E$. Hence for every $X\in\Xloc$, the one form $\Pi\circ \omega(X)$ reads
\[
\Pi \circ\omega(X)= (j^\infty)^*\left(\sum_{i=1}^{\mathrm{rk}(E)}\omega_{ij}X_{u^j} d_Vu^i\right)
\]
On the other hand, a source form $\alpha$ can be written in normal form as \cite{Anderson}
\[
\alpha = (j^\infty)^*\left(\sum_{i=1}^{\mathrm{rk}(E)} \alpha_id_V u^i\right)
\]
hence, the surjectivity of $\Pi\circ (\omega)^\flat$ follows from $\omega_{ij}$ being invertible on $E$.

Then $\Pi_\src\circ (\omega)^\flat\colon \Xloc \to \Omega_{\src}^{1,n}$ is an isomorphism of vector spaces. We show that it is a chain map w.r.t.\ $-\LQ$ and $\wt{d}_Q$, respectively, by showing that 
\[
(\omega)^\flat\colon (\Xloc, -\LQ) \to (\oloc^{1,n}, d_H-\LQ)
\]
is a chain map:
\begin{align*}
\iota_{-[Q,X]}\omega = -\LQ \iota_X\omega - \iota_X\LQ \omega\\
=-\LQ\iota_X\omega - \iota_Xd_H\omega = (d_H - \LQ)\iota_X\omega
\end{align*}
where we used that $0=[\iota_X,d_H] =\iota_X d_H + d_H\iota_X$ and $[\LQ,\iota_X] = \LQ\iota_X + \iota_X\LQ$ for $X$ is of degree $0$ and $Q$ is of degree $1$. Projecting to source forms we get
\[
\Pi\circ(\omega)^\flat(-[Q,X]) = -\Pi\circ \LQ\left((\omega)^\flat(X)\right) = \wt{d}_Q\iota_X\omega,
\]
which then allows us to conclude that $\Pi\circ (\omega)^\flat$ is an isomorphism of complexes, and $(\omega)^\flat$ is a quasi-isomorphism owing to the deformation retract of Lemma \ref{lem:nonhamiltooniandefretract}.
\end{proof}

\begin{remark}[Trivial extension]
\label{Rem: trivialExt}
We want to stress that the construction of the development of 
$\omega$ works for the special case $Q=0$, where we get $\omega^\bullet=\omega$. 
\end{remark}

\begin{remark}
Looking at the definition of Hamiltonian functions and their respective vector fields, we have 
    \begin{itemize}
        \item $\bb{d_H\bF}=0$ for all $\bF\in \ohor^\bullet$ by Theorem \ref{thm:sourcedecomposition},
        \item $\bb{\bF}=0$ for all $\bF$ of horizontal degree less than $n$ and 
        \item for $\bF\in \oham^\bullet$, also $Q(\bF)\in \oham^\bullet$ with $\bb{Q(\bF)}=[Q,\bb{\bF}]$.  Indeed 
        \begin{multline*}
        d_VQ(\bF) = d_V\iota_Q d_V\bF = d_V \iota_Q\iota_{\bb{\bF}}\bom + d_H(\alpha) \\
        = \iota_{[Q,\bb{\bF}]}\bom \pm \iota_{\bb{\bF}}d_V \iota_Q \bom \pm \iota_Q d_V\iota_{\bb{\bF}}\bom + d_H(\alpha) = \iota_{[Q,\bb{\bF}]}\bom + d_H\alpha'
        \end{multline*}
        for some $\alpha,\alpha'\in\oloc^{1,\bullet-1}$, whence $\Pi(d_VQ(\bF)) = \Pi(\iota_{[Q,\bb{\bF}]}\bom)$.
        
    \end{itemize}
As a consequence, we have that  
    \begin{align*}
        \oham^i=
        \begin{cases}
            \oham^n, & \text {for } i=n\\
            \ohor^i, & \text{else}
        \end{cases}
    \end{align*}
\end{remark}

\begin{definition}[Hamiltonian cone. and local Hamiltonian functionals]\label{def:HamCone}
    The \emph{Hamiltonian cone} relative to a $k$-symplectic development $\bom$ is the subspace of the horizontal cone $C^\bullet(p^*)$ given by
    \begin{equation*}
        C_\ham^\bullet 
        = 
        \Omega(M)[n+1]\oplus\oham^\bullet[n]
        \subseteq \Omega(M)[n+1]\oplus\ohor^\bullet[n]
        =C^\bullet(p^*),
    \end{equation*}
    where $\oham^\bullet$ is the space of Hamiltonian forms w.r.t. $\bom$.
    
    Local Hamiltonian functionals are given by the image of the restriction of $\bbP$ to the Hamiltonian cone: $\mathcal{F}_\ham\doteq\im(\bbP |_{C_\ham})$.
\end{definition}

\begin{lemma}
    If $\omega$ is ultralocal then $\oham^\bullet=\ohor^\bullet$. Then also $C_\ham^\bullet=C^\bullet(p^*)$ and $\calF_\ham=\calF$.
\end{lemma}
\begin{proof}
    The source form $d_V\bF$ is in the image of the isomorphism $\Pi\circ(\omega)^\flat$ when $\omega$ is ultralocal due to Theorem 
    \ref{thm:extomega}. 
\end{proof}

\begin{lemma}
    The Hamiltonian cone is a deformation retract of the local Hamiltonian functions:
    \begin{equation}\label{e:Hamiltonianconeretract}
        \xymatrix{
          H \circlearrowright (C^\bullet_\ham, D)\ar@<1ex>[r]^-{\bbP } & \mathcal{F}_\ham[n] \ar@<1ex>[l]^-{i}
        },
    \end{equation}
    where we use the same symbols for the restricted maps. (Note that we may occasionally omit the inclusion map $i(\ell) = (0,\ell)$ for ease of notation).
\end{lemma}
\begin{proof}
    The aim is to check that the maps $D$, $H$ and $i\circ\bbP\equiv \mathbb{p}$ from Theorem \ref{thm:resolution} all restrict to the Hamiltonian cone. Note that the inclusion is an equality for all horizontal degrees less than $0$ and thus $H$ preserves $C_\ham^\bullet$ already, because of degree resons. Hence, let $\mathbb{c}=(\alpha, \bF)\in (\Omega^\bullet(M)[n+1]\oplus \ohor^\bullet[n])^{-1}$, then 
    \begin{align*}
        D\mathbb{c}=D(\alpha,\bF)=(0, d_H\bF+p^*\alpha),
    \end{align*}
    which leaves us to check that $d_H\bF + p^*\alpha\in (\oham^\bullet[n])^{0}$. We have 
    \begin{align*}
        \Pi(d_V(d_H\bF+p^*\alpha))=\Pi(d_Vd_H\bF)=
        -\Pi(d_Hd_V\bF)=0
    \end{align*}
    and therefore $d_H\bF + p^*\alpha$ is Hamiltonian with vector field $0$. 
    Let $(0,\bF)\in (C_\ham^\bullet[n])^0$, then we have that
    \begin{align*}
        i \circ \bbP (0,\bF)=(0,\bF)-DH(0,\bF).
    \end{align*}
    By assumption $\bF\in \oham^n$ and the image of $D$ lies in $C_\ham^{0}$ so also $i\circ \bbP $ preserves $C_\ham^\bullet$. The claimed deformation retract follows, canonically.
\end{proof}

Consider now the differential $D-\LQ$ on $C_\ham^\bullet$ (trivially acting on $\Omega^\bullet(M)[n+1]$).

\begin{lemma}\label{lem:pert.ham.def.ret}
    There is a deformation retract
    \[   
    \label{e:LQdefretract}
    \xymatrix{
      \wt{H}_Q \circlearrowright (C^\bullet_\ham, D-\LQ)\ar@<1ex>[r]^-{\wt{\bbP }} & (\mathcal{F}_\ham[n], d_\ham)\ar@<1ex>[l]^-{{\wt{i}}}
    },
    \]
    with
    \begin{itemize}
        \item $\wt{\bbP }=\bbP \circ
        (\sum_{k\geq 0}(\LQ H)^k)=\bbP +\bbP H\sum_{k\geq 0}(\LQ H)^k)=\bbP $
        \item $\wt{i}= (\sum_{k\geq 0}(H\LQ)^k)\circ i$
        \item $\wt{H}_Q=H\sum_{k\geq 0}(\LQ H)^k)$
        \item $d_\ham=-\bbP\LQ\wt{i}=-\bbP \LQ i$
    \end{itemize}
    Moreover, the retract restricts to
    \[
    \xymatrix{
      \wt{h}_Q \circlearrowright (\Omega_{\ham,0}^\bullet, d_H-\LQ)\ar@<1ex>[r]^-{\bbP } & (\mathcal{F}_\ham[n], \wt{d})\ar@<1ex>[l]^-{{\wt{i}}},
    }
    \]
where   $\Omega_{\ham,0}^\bullet=\oham^\bullet\cap \ker 0^*$.
\end{lemma}
\begin{proof}
    We consider now the deformation retract from Equation \eqref{e:Hamiltonianconeretract}. Let us first shift it by $[n]$
    \[   
        \xymatrix{
          H \circlearrowright (C^\bullet_\ham, D)\ar@<1ex>[r]^-{\bbP } & \mathcal{F}_\ham[n] \ar@<1ex>[l]^-{i}, 
        }
    \]
    where we denote by $i\colon \calF_\ham[n]\hookrightarrow C_\ham^\bullet$ the inclusion map, and then perturb it in the direction of $-\LQ$ to obtain the new claimed deformation retract.
    
    Note that in the first point we used that $\bbP H=0$ and in the last point that $\bbP \circ\wt{i}=\bbP \circ i=\id$, because of degree reasons. The perturbation $-\LQ$ restricts to $\Omega_{\ham,0}^\bullet\hookrightarrow C_\ham^\bullet$, since
    \begin{align*}
        0^*\LQ=0,
    \end{align*}
    and given that all other maps also restrict due to Theorem \ref{Thm: DefRetractwithoutconstants} we prove the claim.
\end{proof}

\subsection{Strong-homotopy data}

In finite dimensional symplectic geometry, a symplectic form $\Omega$ on $M$ induces a Poisson bracket on the space of functions $C^\infty(M)$. Since the only calculus one needs to worry about in that case is the Cartan calculus on $M$, we have 
\[
\{f,g\}_\Omega \doteq \iota_{X_f}\iota_{X_g}\Omega = L_{X_f}(g) = -L_{X_g}(f) = \frac12 \left(L_{X_f}(g) - L_{X_g}(f)\right).
\]
When dealing with local symplectic forms on the variational bicomplex (and their developments) we see that the above expressions for the Poisson bracket will differ.

\begin{proposition}\label{prop:bracket}
    Let $\bom$ be a symplectic development of a $k$-symplectic form $\omega$, and let $\oham^\bullet$ be the associated space of Hamiltonian forms. 
    Then there are the following antisymmetric bilinear brackets of degree 
    $-k$:
    \begin{align}
        \{\cdot,\cdot\}^S&\colon \oham^\bullet\times \oham^\bullet \to \oham^\bullet &   \{\bF,\bul{G}\}^S&\doteq\iota_{\bb{\bF}}\iota_{\bb{\bul{G}}}\bom\\
        \{\cdot,\cdot\}^A&\colon \oham^\bullet\times \oham^\bullet \to \oham^\bullet & \{F^\bullet,G^\bullet\}^A&\doteq \frac12\left(L_{\bb{\bF}}(\bul{G})-\sigma_k(\bul{F},\bul{G}) L_{\bb{\bul{G}}}(\bF)\right),\\
        \{\cdot,\cdot\}^B&\colon \oham^\bullet\times \oham^\bullet \to \oham^\bullet & \{F^\bullet,G^\bullet\}^B&\doteq 2\{\bul{F},\bul{G}\}^A - \{\bul{F},\bul{G}\}.
    \end{align}
where $\sigma_k(\bul{F},\bul{G})=(-1)^{(\mathsf{ped}(\bF)-k)(\mathsf{ped}(\bul{G})-k)}$ and the skew-symmetry is with respect to the partial effective degree. Moreover
\[
\bb{\{\bul{F},\bul{G}\}^S}=\bb{\{\bul{F},\bul{G}\}^A}=\bb{\{\bul{F},\bul{G}\}^B}=[\bb{\bul{F}},\bb{\bul{G}}],
\]
and for $\star,\star'\in \{S,A,B\}$
\[
\{\bul{F},\bul{G}\}^\star-\{\bul{F},\bul{G}\}^{\star'} \in d\oloc^{0,\top-1}\oplus \oloc^{1,<\top}.
\]
\end{proposition}

\begin{proof}
We begin by observing that the assignment $F\to \bb{\bul{F}}$ is well defined, since 
$\Pi\circ(\bom)^\flat
=\Pi\circ (\omega)^\flat$ is injective. Moreover one can check that all brackets have codomain $\oham^\bullet$ and that 
\[
\bb{\{\bul{F},\bul{G}\}^S}=\bb{\{\bul{F},\bul{G}\}^A}=\bb{\{\bul{F},\bul{G}\}^B}=[X_F,X_G]
\]
by usual (graded) Cartan calculus, and that 
\[
\Pi\left(d_VL_{\bb{\bL}}(\bb{\bul{G}})\right)= \Pi\left(d_V\iota_{\bb{\bul{L}}}d_V \bul{G}\right) = 
\begin{cases}
    0 & G^0 = 0\\
    \Pi\left(d_V\iota_{\bb{\bul{L}}}\iota_{\bb{\bul{G}}}\bom\right) & G^0\not= 0
\end{cases}
\]
where we used that, when $G^0\not=0$, $\iota_{\bb{\bul{G}}}\bom = d_V \bul{G} + d_H\alpha$ for some $\alpha\in\oloc^{1,\bullet}$, and both $d_V$ and $\iota_{\bb{\bF}}$ commute with $d_H$. Thus
\[
\Pi\left(d_VL_{\bb{\bL}}(\bb{\bul{G}})\right)= \Pi(d_V\{\bul{F},\bul{G}\}) = \Pi(\iota_{[\bb{\bul{F}},\bb{\bul{G}}]}\bom).
\]
Then, it is clear that $L_{\bb{\bL}}(\bb{\bul{G}}) - \iota_{\bb{\bul{L}}}\iota_{\bb{\bul{G}}}\bom$ will be $d_H$ exact if $G^0\not=0$ and will have generically components in non top degree, hence
\[
\{\bul{F},\bul{G}\}^A - \{\bul{F},\bul{G}\}^S \in d\oloc^{0,\top-1}\oplus \oloc^{1,<\top}.
\]
and similarly for $\{\cdot,\cdot\}^B$.
\end{proof}

\begin{definition}
    Let us extend the brackets $\{\cdot,\cdot\}^\star$ of Proposition \ref{prop:bracket} to $C_\ham^\bullet$ by 
    \begin{align*}
        \{\mathbb{c}_1,\mathbb{c}_2\}^\star=\{(\alpha_1,\bF_1),(\alpha_2,\bF_2)\}=
        (0,\{\bF_1,\bF_2\}^\star), \qquad \forall \mathbb{c}_1,\mathbb{c}_2\in C_\ham^\bullet.
    \end{align*}
We call $\{\cdot,\cdot\}$ the \emph{standard bracket}, $\{\cdot,\cdot\}^A$ the \emph{antisymmetrised bracket}, and $\{\cdot,\cdot\}^B$ the \emph{Bonechi--Chevalley--Courant--Getzler--Soloviev}\footnote{See Remark \ref{rmk:naming}.} bracket associated to $\bom$.
\end{definition}

\begin{lemma}
\label{Lem: Jacexact}
    Let $\bom$ be a $(d_H-\LQ)$-closed symplectic development of $\omega$ (cf.\ Theorem \ref{thm:extomega}). The differential $D-\LQ$ is a derivation of both the standard and the antisymmetrised brackets,\footnote{Recall that the brackets are defined on the Hamiltonian cone, where $D$ acts, as opposed to $d_H$ which acts on its summand $\ohor$.} and thus also of $\{\cdot,\cdot\}^B$.
\end{lemma}
\begin{proof}
The derivation property for $\{\cdot,\cdot\}^S$ follows from that fact that $(d_H-\LQ)\omega=0$  and that if a $\bF\in \Omega_\ham^\bullet$ with associated vector field $\bb{\bF}$, then the vector field of $Q(\bF)$ is given by $[Q,\bb{\bF}]$. For the antisymmetrised bracket, we use this fact together with $[L_{\bb{\bF}},d_H]=0$ to get that 
$D-\LQ$ is a derivation of $\{\cdot,\cdot\}^A$ (in fact even $D$ and $\LQ$ separately). 
\end{proof}


\begin{lemma}
\label{Lem: PkillsJaco}
The Jacobiators $\mathrm{Jac}_{\{\cdot,\cdot\}^S}$ and $\mathrm{Jac}_{\{\cdot,\cdot\}^A}$ are in the kernel of $\mathbb{p}$. Moreover, the Jacobiator $\mathrm{Jac}_{\{\cdot,\cdot\}^{B}}$ identically vanishes, turning $(C_\ham^\bullet,\{\cdot,\cdot\}^B)$ into a Lie algebra.
\end{lemma}
\begin{proof}
    We observe that the result of $\mathrm{Cyc}_{\bF,\bul{G},\bul{H}}(\{\bF,\{\bul{G},\bul{H}\}^S\}^S)$ has vanishing associated vector field $\mathbb{0}$ for any $F^\bullet, G^\bullet, H^\bullet \in \Omega_\ham^\bullet$, since it is given by the Jacobiator of $\bb{\bul{F}},\bb{\bul{G}},\bb{\bul{H}}$, which vanishes. This means that 
        \begin{align*}
            \Pi_\src(d_V\pi_2\mathrm{Jac_{\{\cdot,\cdot\}^S}}((0,F^\bullet), (0,G^\bullet), (0,H^\bullet)))=
            \Pi_\src (\iota_{\mathbb{0}}\omega^\bullet)=0. 
        \end{align*}
    Using the definition of $\mathbb{p}=(P_V\circ \wt{I})\circ(\Pi_\src\circ I_V)$, and that $\bbP(\{\cdot,\cdot\}^S-\{\cdot,\cdot\}^A)=0$ (Proposition \ref{prop:bracket}) we get the first claim. 

    To show that $\mathrm{Jac}_{\{\cdot,\cdot\}^{B}}=0$ we observe that $\Xloc$ acts on $\ohor^\bullet$ and thus the semidirect product $\Xloc\ltimes \ohor^\bullet$ is a Lie algebra with $[(X,\bul{F}),(Y,\bul{G})] = ([X,Y],X(\bul{G}) - Y(\bul{F}))$. Then, for any vertically closed form $\bom\in\oloc^{2,\bullet}$, the central extension 
    \[
    [(X,\bul{F}),(Y,\bul{G})]_{\bom} = ([X,Y],X(\bul{G}) - Y(\bul{F}) - \iota_{X}\iota_Y\bom)
    \]
    is a Lie algebra. Then, we observe that this bracket restricts to Hamiltonian forms $\oham$ and we immediately conclude that
    \[
    [(\bb{\bul{F}},\bul{F}),(\bb{\bul{G}},\bul{G})]_{\bom} = ([\bb{\bul{F}},\bb{\bul{G}}],\{\bul{F},\bul{G}\}^B).
    \]
    
\end{proof}

\begin{remark}[Nomenclature]\label{rmk:naming}
    Let us explain our choice of nomenclature for $\{\cdot,\cdot\}^B$. Chevalley is attributed due to it being a Lie bracket as a consequence of $\bom$ being a central extension of the Lie algebra on $\Xloc\ltimes \ohor^\bullet$, while Courant is attributed because $\mathfrak{X}(M)\oplus C^\infty(M)$ can be thought of as a $0$-Courant algebroid, see \cite{Zambon2012}. This crucial observation was communicated to us by F.\ Bonechi, hence the reference. 

    Soloviev is credited for having introduced a (coordinate dependent) lift of the BV bracket on $\calF_\ham$ to $\ohor$ as a Lie bracket \cite{Soloviev}---a property shared by our bracket $\{\cdot,\cdot\}^B$---and E.\ Getzler is credited for repopularising \cite{GetzlerDarboux} the Soloviev bracket. Finally, we observe that this bracket was used by Barnich and Trossaert in a similar way when analysing charge algebras on certain reduced phase spaces for field theory \cite{BarnichTroessaert} (which corresponds to our $k=0$, $Q=0$ case).
\end{remark}


\begin{theorem}
\label{Thm: Linfty}
    The degree $-k$ bracket on $\mathcal{F}_\ham [n]=\im(\bbP \vert_{C^\bullet_{\ham}})[n]$ defined by
    \begin{align*}
        \{\cdot,\cdot\}_\ham\colon \mathcal{F}_\ham [n]^{\otimes 2}\ni (\bbP \mathbb{c}_1\otimes \bbP \mathbb{c}_2)\mapsto\bbP \{\mathbb{c}_1,\mathbb{c}_2\}\in \mathcal{F}_\ham[n],
    \end{align*}
together with the map 
    \begin{align*}
        d_\ham \colon \mathcal{F}_\ham [n]\ni \bbP \mathbb{c} \mapsto -\bbP (\LQ \mathbb{c}) \in \mathcal{F}_\ham [n],
    \end{align*}
turn $(\mathcal{F}_\ham[n],d_\ham,\{\cdot,\cdot\}_\ham)$ into a dgL$[k]$a.  Moreover,
\begin{enumerate}
    \item The map
    \begin{align*}
        \{\cdot,\cdot,\cdot\}^S\colon C^{\otimes 3}_\ham \ni (\mathbb{c}_1\otimes \mathbb{c}_2\otimes \mathbb{c}_3)\mapsto 
        -\sum_{k\geq 0}(H\LQ)^kH\mathrm{Jac}_{\{\cdot,\cdot\}}(\mathbb{c}_1,\mathbb{c}_2,\mathbb{c}_3)\in C_\ham^\bullet
    \end{align*}
turns $(C^\bullet_\ham,D-\LQ,\{\cdot,\cdot\}^S,\{\cdot,\cdot,\cdot\}^S)$ into a $L_\infty[k]$-algebra and 
\[
\bbP \colon \left(C^\bullet_\ham,D-\LQ,\{\cdot,\cdot\}^S,\{\cdot,\cdot,\cdot\}^S\right) \to \left(\mathcal{F}_\ham[n],\{\cdot,\cdot\}_\ham\right)
\]
is a (strict)  $L_\infty[k]$-algebra quasi-isomorphism. 
\item The map 
\[
\bbP \colon \left(C^\bullet_\ham,D-\LQ,\{\cdot,\cdot\}^B\right) \to \left(\mathcal{F}_\ham[n],\{\cdot,\cdot\}_\ham\right)
\]
is a dgL$[k]$ algebra quasi-isomorphism. 
\end{enumerate}
    
\end{theorem}
\begin{proof}
We begin by checking that the bracket $\{\cdot,\cdot\}^S_\ham$ is well defined. In order to do this, we change the representative $\mathbb{c} \mapsto \mathbb{c} + (K, d\bul{f} + p^* k)$. Since the bracket on Hamiltonian functionals $\{\cdot,\cdot\}^S$ is extended trivially to the Hamiltonian cone, we only need to check
\[
\{d\bul{f} + p^* k, \bul{F}\} = 0,
\]
but this is obvious since the Hamiltonian vector field of $d\bul{f} + p^* k$ is zero.

Let us then check that $d_\ham$ is a differential. Let $\bbP\mathbb{c} \in \mathcal{F}_\ham$; using the homotopy equation $\id = [D,H] + \mathbb{P}$, we get 
    \begin{align*}
        \bbP (\LQ \bbP (\LQ \mathbb{c}))
        &=\bbP (\LQ(\LQ \mathbb{c}-DH\LQ \mathbb{c}))
        =-\bbP (\LQ DH\LQ \mathbb{c})\\&
        =\bbP (D\LQ H\LQ \mathbb{c})=0,
    \end{align*}
since $\bbP $ vanishes on the image of $D$.
Using Lemma \ref{Lem: PkillsJaco}, we see that $\{\cdot,\cdot\}_\ham$ is in fact a degree $-k$ Lie bracket. 
This turns $(\mathcal{F}[n],d_\ham,\{\cdot,\cdot\}_\ham)$ into a $dgL[k]a$. 
Note that $(D-\LQ,\{\cdot,\cdot\}^S)$ is a $L_\infty[k]$ structure up to order $2$ already, which only means that $D-\LQ$ is differential and a derivation of $\{\cdot,\cdot\}^S$. Morover, $\bbP $ is an $L_\infty[k]$-morphism up to order $2$ by definition. The idea is now to extend this via an explicit version of the homotopy transfer theorem, see \cite[Theorem 4.11]{bestformulaever}: using the explicit formulas and identities we have in our special case (i.e.\ that we do not only have a homotopy equivalence, but a deformation retract), we obtain for the Taylor coefficients of the coderivation $\delta$, where $\delta_1^1$ is induced by the differential $D-\LQ$ and $\delta^1_{2}$ is induced by the bracket, that 
    \begin{align}
    \label{Eq: Definitionbrackets}
        \delta^1_{k+1} = 
        -\wt{H}_Q \sum_{i=2}^k \delta^1_i\delta^i_{k+1}
    \end{align}
completes to a $L_\infty[k]$-algebra. 
Moreover, there are no higher corrections of the map 
$\bbP $. Using now that the bracket vanishes on the image of $\wt{H}_Q$, we conclude that for $k\geq 2$
    \begin{align*}
     \delta^1_{k+1}=(-\wt{H}_Q)^{k-1}(\delta_2^1\delta^2_3\cdots \delta^{k}_{k+1})   
    \end{align*}
which implies that for $k\geq 4$ that $\delta_k^1=0$, since $\delta_3^2\delta_4^3$ is the Jacobiator of $\{\cdot,\cdot\}^S$ (up to constant factor) and therfore is in the kernel of $\mathbb{P}$ (see Lemma \ref{Lem: PkillsJaco}), Using the homotopy equation for $D$, we get that the Jacobiator is in the span of the images of $H$ and $D$, and the bracket $\{\cdot,\cdot\}^S$ vanishes on there. Moreover, this map can be identified via the décalage isomorphism with  
 $$\{\cdot,\cdot,\cdot\}^S\colon C_\ham^{\otimes 3} \ni (\alpha\otimes\beta\otimes \gamma)\mapsto 
        -\wt{H}_Q\mathrm{Jac}_{\{\cdot,\cdot\}}(\alpha,\beta,\gamma)\in C.$$
 Alternatively, one can check by hand that  $(C_\ham,D-\LQ,\{\cdot,\cdot\}^S,\{\cdot,\cdot,\cdot\}^S)$ is a $L_\infty[k]$ algebra and that $\bbP $ is a strict $L_\infty[k]$-morphism and since it is part of the deformation retract \eqref{e:LQdefretract}, it is a quasi-isomorphism.  

From Proposition \ref{prop:bracket} and from the first part of the proof, we know that 
    \begin{align}
        \mathbb{P}\{\bul{F},\bul{G}\}^B= \mathbb{P}\{\bul{F},\bul{G}\}^S=\{\mathbb{P}\bul{F},\mathbb{P}\bul{G}\}_\ham
    \end{align}
and thus $\mathbb{P}$ is a morphism of graded Lie algebras with respect to the Lie bracket $\{-,-\}^B$ and the claim is proven. 
\end{proof}

\begin{remark}
    Using Remark \ref{Rem: Stasheff}, we see that in case the manifold is contractible, we can obtain an $L_\infty$ structure on $\mathbb{R}[1]\oplus \ohor^\bullet$. A comparison between this $L_\infty$ data and the one constructed by \cite{BFLS} will be given in Remark \ref{rmk:comparison}.
\end{remark}

Using the homotopies and that $\bbP $ is already an $L_{\infty}[k]$ morphism, we can compute an explicit quasi-inverse.

\begin{lemma}\label{lem:quasiinverse}
The sequence of maps $\tilde{\mathbb{I}}_{k}^1\colon S \mathcal{F}[k+1]^\bullet \to C^\bullet_\ham$ inductively defined by $\wt{\mathbb{I}}^1_1=\wt{i}$ and 
    \begin{align*}
        \wt{\mathbb{I}}^{1}_{k+1}=-\wt{H}_Q L_{\infty,{k+1}}(\wt{\mathbb{I}}),
    \end{align*}
where $L_{\infty,{k+1}}(\wt{\mathbb{I}})=\sum_{l=2}^{k+1}\delta_{l}^1\wt{\mathbb{I}}_{k+1}^l- \sum_{l=1}^k\wt{\mathbb{I}}_{l}^1\delta_{\mathcal{F},k+1}^l$ is an $L_\infty$-quasi-inverse of $\bbP $ with the structure maps $Q$ and $\delta_\mathcal{F}$ of the corresponding (shifted) $L_\infty$-structures, where the $\delta^1_k$ are the Taylor coeffcients of the $L_\infty[k]$ structure of $\left(C^\bullet_\ham,D-\LQ,\{\cdot,\cdot\}^B\right)$ or $\left(C^\bullet_\ham,D-\LQ,\{\cdot,\cdot\}^S,\{\cdot,\cdot,\cdot\}^S\right)$. We call the associated quasi-inverses $\mathbb{I}^B$ and $\mathbb{I}^S$ respectively.
\end{lemma}

\begin{proof}
We are going to prove the claim inductively. Assume that $\wt{\mathbb{I}}$ is an $L_\infty$-
morphism up to order $k$. Using \cite[Lemma 4.3]{kraft:2024}, we know that 
    \begin{align*}
        \delta_1^1L_{\infty,k+1}(\wt{\mathbb{I}})
        =L_{\infty,k+1}(\wt{\mathbb{I}})\delta^{k+1}_{\mathcal{F},k+1}.
    \end{align*}
Then we know that 
    \begin{align*}
        \wt{\mathbb{I}}^1_{k+1}\delta_{\mathcal{F},k+1}^{k+1} - \delta_{1}^1\wt{\mathbb{I}}^1_{k+1}&= 
		\delta_{1}^1\circ \wt{H}\circ L_{\infty,{k+1}}(\wt{\mathbb{I}})
		-\wt{H}\circ L_{\infty,{k+1}}(\wt{\mathbb{I}})\circ \delta_{\mathcal{F},k+1}^{k+1}\\&
		=\delta_{1}^1\circ \wt{H}\circ L_{\infty,{k+1}}(\wt{\mathbb{I}})+\wt{H}\circ \delta_{1}^{1}\circ L_{\infty,{k+1}}(\wt{\mathbb{I}})\\&
		=(\mathrm{id}-\iota\circ \bbP )L_{\infty,{k+1}}(\wt{I}).
    \end{align*}
Using that $\bbP \circ \wt{H}=0$, we see already that $\bbP \circ \wt{\mathbb{I}}^1_k=0$ for $k\geq 2$, using this and that $\bbP $ is an $L_\infty$-morphism implies that $\bbP L_{\infty,{k+1}}(\wt{I})=0$. Moreover, 
we know that $\bbP \circ\wt{\mathbb{I}}=\mathrm{id}$, which is enough to show that $\wt{\mathbb{I}}$ is the quasi-inverse. 
\end{proof}

\begin{remark}
Note that there are only finitely many non-vanishing $\mathbb{I}^1_k$, since each of the $\mathbb{I}^1_{k}$ contain at least $k-1$ horizontal homotopies and only operations which do not increase the horizontal degree.     
\end{remark}

Let us finish this subsection with some observations and let us restrict ourselves to the case of the standard bracket first. If we apply part one of Theorem \ref{Thm: Linfty} for the case $Q=0$ and $\omega^\bullet=\omega$ (see Remark \ref{Rem: trivialExt}), then we get a diagram of $L_\infty$-algebras 
\begin{equation}\label{e:Stasheffsres}
        \xymatrix{
         (C^\bullet_{\ham,0}, D,\{\cdot,\cdot\}^S_0,\{-,-,-\}^S_0)\ar@<1ex>[r]^-{\bbP } & (\mathcal{F}_\ham[n], \{\cdot,\cdot\}_\ham)
         \ar@<1ex>[l]^-{\widetilde{\mathbb{I}}_0^S} \ar@<1ex>[r]^-{\widetilde{\mathbb{I}}_0^B} & \ar@<1ex>[l]^-{\bbP }(C^\bullet_\ham, D, \{\cdot,\cdot\}^B_0)
        },
    \end{equation}
where $\bbP$ and $\widetilde{\mathbb{I}}^{S/B}_0$ are the quasi-isomorphisms (quasi-inverse to each other) constructed in Lemma \ref{lem:quasiinverse} and $\{\cdot,\cdot\}^S_0$ and $\{\cdot,\cdot\}^B_0$ are the brackets specialised for the data $Q=0$ and $\omega^\bullet=\omega$. Note that $(\calF_\ham,\{\cdot,\cdot\}_\ham)$ has now vanishing differential. 

Let us now assume that there is a functional $L^\bullet$ with\footnote{In other words a Hamiltonian pair $(\bL,Q)$ w.r.t.\ $\omega$.} $\Pi(\iota_Q\omega-d_V\bL)=0$, where $Q$ is now allowed to be non-zero, but $Q^2=0$ and such that $\Pi(\LQ\omega)=0$. 

We claim now that 
$\ell \doteq \bbP (0,\bL)$ is a Maurer-Cartan element, since\footnote{Note that this calculation could have been done using $\{\cdot,\cdot\}^B$, but this version is manifestly simpler.} 
    \begin{align*}
     \{\ell,\ell\}_\ham &=\bbP (0,\{L^\bullet,L^\bullet\}^S_0)\\&
     =\Pi d_V\iota_Q\iota_Q\omega
     =\Pi d_V \iota_Q (d_V \bL +d_H X)
     =\Pi(d_VQ(\bL))
    \end{align*}
where we used that $\Pi(\iota_Q\bom-d_V\bL)=0$ and therefore we have that, using the horizontal homotopy equation, that $\iota_Q\bom-d_V\bL=d_HX+Y$, where $Y$ has horizontal degree smaller than $\top$ and $\Pi(d_V\iota_Q Y)=0$. 
With the same argument, one checks that 
    \begin{equation*}
        \Pi(d_VQ(\bL))=\Pi(\LQ\bom)=0
    \end{equation*}
and thus $\{\ell,\ell\}_\ham=0$. 

Since $\ell$ is a Maurer--Cartan element in $(\mathcal{F}_\ham[n], \{\cdot,\cdot\}_\ham)$, we can twist the Lie algebra structure via $\ell$ to obtain the new dgL$[k]$a structure
    \begin{align*}
    (\mathcal{F}_\ham[n], \{\ell,-\}_\ham, \{\cdot,\cdot\}_\ham)    
    \end{align*}
Let us first notice that, for $\mathbb{c} = (\alpha,\bul{F})\in C_\ham^\bullet$ 
    \begin{align*}
        \{\ell, \bbP \mathbb{c}\}_\ham &=
        \bbP (0,\{\bL, \bF\}^S_0)\\
        &=h_V\circ I\circ \Pi(d_V \iota_Q d_V(\bF))
        =\bbP (0,\LQ(\bF)) \\
        &= \bbP \circ \LQ \circ \wt{i} (\bul{F}) = d_\ham \bF
    \end{align*}
and therefore $d_\ham=\{\ell,-\}_\ham$.

This implies that 
\begin{proposition}
\label{Prop: twisteq}
The diagram
    {\small\[   
        \xymatrix{
         (C^\bullet_{\ham,0})^{\mathbb{I}^S_{0,MC}(\ell)}\ar@<1ex>[r]^-{\bbP } & \ar[d]^{=}\left(\mathcal{F}_\ham[n],d_\ham, \{\cdot,\cdot\}_\ham\right)
         \ar@<1ex>[l]^-{\mathbb{I}_0^{S,\ell}} \ar@<1ex>[r]^-{\tilde{\mathbb{I}}^S}&
         \left(C^\bullet_{\ham}, D_H - \LQ, \{\cdot,\cdot\}^S, \{\cdot,\cdot,\cdot\}^S \right) \ar@<1ex>[l]^-{\bbP }\\
         (C^\bullet_{\ham,0})^{\mathbb{I}^B_{0,MC}(\ell)}\ar@<1ex>[r]^-{\bbP } & \left(\mathcal{F}_\ham[n],d_\ham, \{\cdot,\cdot\}_\ham\right)
         \ar@<1ex>[l]^-{\mathbb{I}_0^{B,\ell}} \ar@<1ex>[r]^-{\tilde{\mathbb{I}}^B}&
         \left(C^\bullet_{\ham}, D_H - \LQ, \{\cdot,\cdot\}^B \right) \ar@<1ex>[l]^-{\bbP }
         ,
        }
    \]}
    where $\mathbb{I}_0^{S/B,\ell}$ are the $L_\infty[k]$-morphism $\mathbb{I}_0^{S/B}$ twisted by $\ell$ (Equation \eqref{e:twistedmorphism}), is a diagram of quasi-isomorphisms of $L_\infty$-algebras, where $(C^\bullet_{\ham,0})^{\mathbb{I}^S_{0,MC}(\ell)}$ is
    {\small\[
    \left(C^\bullet_{\ham,0}, D + \{\mathbb{I}^S_{0,MC}(\ell), \cdot\}^S_0 + \{\mathbb{I}^S_{0,MC}(\ell),\mathbb{I}^S_{0,MC}(\ell), \cdot\}^S_0, \{\cdot,\cdot\}^S_0 + \{\mathbb{I}^S_{0,MC}(\ell), \cdot, \cdot\}^S_0, \{\cdot,\cdot,\cdot\}^S_0\right)
    \]}
    and $(C^\bullet_{\ham,0})^{\mathbb{I}_{0,MC}(\ell)}$ is
    {\small\[
    \left(C^\bullet_{\ham,0}, D + \{\mathbb{I}^B_{0,MC}(\ell), \cdot\}^B_0, \{\cdot,\cdot\}^B_0\right).
    \]}
\end{proposition}

\begin{remark}[Different, but related constructions]
${}$
    As already discussed in Proposition \ref{prop:bracket}, we could have chosen the alternative degree $-k$ bracket 
    $\{\cdot,\cdot\}^A$, which project to $\{\cdot,\cdot\}_\ham$ via $\bbP$ and therefore can be extended to $L_\infty[k]$-algebras 
    on $C_\ham^\bullet$ with exactly the same proof as for Theorem \ref{Thm: Linfty}. 
    
    A possibly relevant difference consists in the existence of higher brackets than three, since  $\{\cdot,\cdot\}^A$ does not vanish on the image of $\wt{H}_Q$. In fact, every application of $\wt{H}_Q$ lowers the horizontal degree at least by one see Formula \eqref{Eq: Definitionbrackets}, which means that $\delta_k^1$ lowers the horizontal degree at least by $k-1$. Since the horizontal degree is bounded from above and below, there are at most $\dim(M)+1$ non-vanishing brackets, so in particular finitely many. 
    This $L_\infty[k]$-algebra is also quasi-isomorphic to 
    $\left(\mathcal{F}_\ham[n],d_\ham, \{\cdot,\cdot\}^S_\ham\right)$ and therefore also to the one we constructed via the bracket $\{\cdot,\cdot\}^S$. 

In view of Theorem \ref{Thm: DefRetractwithoutconstants} and Lemma \ref{lem:pert.ham.def.ret}, one can also induce an $L_\infty$-algebra structure on $\Omega_{\ham,0}^\bullet$. As explained there, all the maps from the deformation retracts restrict to this subcomplex, but neither $\{\cdot,\cdot\}^S$ nor $\{\cdot,\cdot\}^A$ or $\{\cdot,\cdot\}$ do, because two nonconstant functions can have a constant bracket. Nevertheless, one can use the general homotopy transfer theorem to induce an $L_\infty$-algebra structure on $\Omega_{\ham,0}^\bullet$, which will be $L_\infty$-quasi-isomorphic to all the structures we constructed. However, we will not 
pursue this approach in this work. 

\end{remark}

\begin{remark}[Interpretation]\label{rmk:comparison}
    The Hamiltonian cone for $Q=0$ is nothing but a resolution of the space of local Hamiltonian functionals $\calF_\ham$, which is equipped with the degree $-k$ bracket $\{\cdot,\cdot\}_\ham$. The cone is then equipped with an $L_\infty$ structure as clarified by Equation \eqref{e:Stasheffsres}. This generalises (and coincides up to equivalences with) the strong homotopy Lie algebra constructed by \cite{BFLS} for the case $k=0$ and when $M$ is contractible. When $k=-1$ this is an $L_\infty$ resolution of the BV bracket on local functionals, see Section \ref{sec:BVBFV}.
    
    It is important to note that this data does not come equipped with a ``choice of theory'' in the form of a Lagrangian, which instead is encoded in the Maurer--Cartan element $\ell = \bbP(0,\bL)$. Such a choice results in the data $(C^\bullet_{\ham,0})^{\mathbb{I}^S_{0,MC}(\ell)}$, which is a twist, and can be easily computed to give
    \[
    (C^\bullet_{\ham,0})^{\mathbb{I}^S_{0,MC}(\ell)} = \left((C^\bullet_{\ham,0}, \{\bL, \cdot\}^S_0 + \{\bL,\bL, \cdot\}^S_0, \{\cdot,\cdot\}^S_0 + \{\bL, \cdot, \cdot\}^S_0, \{\cdot,\cdot,\cdot\}^S_0\right)
    \]
    since there are no higher brackets, and where we used that to compute the brackets we can take any element $\bL$ in the preimage of $\ell$ along $\bbP$ (as opposed to $I(\ell)$), since $\{\cdot,\cdot\}^S$ vanishes on the kernel of $\bbP$. Note that this data uses only $\omega$ as input symplectic data, but has more complicated bracket structure. The equivalence stated above tells us that another equivalent model is given by the Hamiltonian cone $C_{\ham}^\bullet$ constructed for a symplectic development $\bom$ of $\omega$, with input $Q$---the unique Hamiltonian vector field of \emph{any} representative of the MC element $\ell$.
\end{remark}

\subsection{Hamiltonian triples}
Recall Definition \ref{def:Hamiltonian} of a Hamiltonian function and a Hamiltonian pair. Given a $k$-symplectic local form $\omega$ and a Hamiltonian form $F$, it is clear that every development $\bF$ of $F$ (Definition \ref{def:developments}) satisfies the same Hamiltonian condition w.r.t.\ any symplectic development $\bom$ of $\omega$, since being Hamiltonian is truly a property that only depends on the $\top$-component $F^0=F$ and the (source part of the) top component $\iota_{\bb{F}}\omega$. 

We have seen in Theorem \ref{thm:extomega} that some $K$-symplectic developments $\bom$ can be built out of the data $\omega$ and $Q$, and we will investigate now how certain forms satisfy stronger Hamiltonian properties with respect to those developments.\footnote{Note that the following definition can be generalised to generic developments of a local symplectic form, i.e.\ not necessarily $d_V$-closed ones.}

\begin{definition}[Hamiltonian triples]\label{def:hamtriples}
    Let $\bom\in \oloc^{2,\bullet}$ be a $k$-symplectic development of $\omega$.
    We say that a Hamiltonian pair $(\bF,\bb{\bF})$ completes to a Hamiltonian triple $(\bF,\bb{\bF},\bth_{\bF})$ w.r.t.\ $\bom$ iff there exist $\bth_{\bF}\in\oloc^{1,\bullet}(\mathcal{E}\times M)$ such that
    \[
    \iota_{\bb{\bF}}\bom =  d_V  \bF + d_H\bth_{\bF}.
    \]
    We denote by $\mathbb{HT}_\omega\subset \oloc^{0,\bullet}\oplus \Xloc \oplus \oloc^{1,\bullet}$ the subspace of Hamiltonian triples. Moreover, the space of Hamiltionian functionals that complete to a local Hamiltonian triple will be denoted with $\oht^\bullet\subset \oham^\bullet$. The associated cone of Hamiltonian triples will be denoted by $C_{\mathrm{ht}}^\bullet\doteq\Omega^\bullet(M)[1]\oplus \oht^\bullet$.
\end{definition}

\begin{remark}[Starting datum for $\bth$]
    Observe that the Hamiltonian triple equation involves all terms of the inhomogenous form $\bom$ except the $\top$ term $\theta^0$, which is not determined by the Hamiltonian triple condition. We find it useful to adopt a convention by which we fix $\theta^0=h_V\omega$.
\end{remark}

Note that, obviously $\oht^\bullet\subseteq \oham^\bullet$ by definition. On the other hand if $\bF\in \oham^\bullet$, then we get 
    \begin{align}
        \iota_{\bb{\bF}}\bom -d_V \bF&=
        (d_Hh^\nabla +h^\nabla d_H+ \Pi)
        (\iota_{\bb{\bF}}\bom -d_V \bF)\\&
        =
        (d_Hh^\nabla +h^\nabla d_H)
        (\iota_{\bb{\bF}}\bom -d_V \bF)
    \end{align}
This means that $\bF$ can be extended to a Hamiltonian triple if and only if 
\[
d_H(\iota_{\bb{\bF}}\bom -d_V \bF)=0, 
\]
in which case we can define 
\[
\bth\doteq h^\nabla(\iota_{\bb{\bF}}\bom -d_V \bF).
\]

\begin{remark}
    For Hamiltonian triples we can give a more precise relation between the two brackets on Hamiltonian functionals defined above. Suppose $\bth_{\bul{F}},\bth_{\bul{G}}\in \oloc^{1,\top-1}$ are such that
    \[
    \iota_{\bb{\bul{F}}}\bom = d_V \bul{F} +d_H \bth_{\bul{F}}, \qquad \iota_{\bb{\bul{G}}}\bom = d_V \bul{G} +d_H \bth_{\bul{G}},
    \]
    then
    \[
    \{\bul{F},\bul{G}\}^A = \{\bul{F},\bul{G}\} - \frac12d_H\left(\iota_{\bb{\bul{F}}}\bth_{\bul{F}} - \sigma_k(\bul{F},\bul{G})\iota_{\bb{\bul{G}}}\bth_{\bul{F}}\right),
    \]
    and similarly
    \[
    \{\bul{F},\bul{G}\}^B = \{\bul{F},\bul{G}\} - d_H\left(\iota_{\bb{\bul{F}}}\bth_{\bul{F}} - \sigma_k(\bul{F},\bul{G})\iota_{\bb{\bul{G}}}\bth_{\bul{F}}\right),
    \]
\end{remark}

Let us now consider the homological vector field $Q$. We can consider the relationship between a $k$-symplectic development $\omega^\bullet$ that is $(d_H - \LQ)$-closed, such as the one constructed in Theorem \ref{thm:extomega}, and the existence of 
Hamiltonian triples $(\bL,Q,\bth)$.

\begin{theorem}[Classification of canonical triples]\label{thm:Hamtriples}
    Let $\bom$ be a $k$-symplectic development of a local $k$-symplectic form $\omega$, and let $(\bL,Q)$ be a Hamiltonian pair w.r.t.\ $\bom$, such that $Q\in\Xloc(\mathcal{E})^1$ is homological. Then, the following are equivalent
    \begin{enumerate}
        \item $\bom$ is $(d_H-\LQ)$-closed,
        \item $(\bL ,Q,\bth)$ is a Hamiltonian triple w.r.t.\ $\bom$, with 
        \begin{equation}\label{e:canonicaltriple}
        \bth \doteq h_V \bom, \qquad \bL \doteq h_V(\iota_Q\bom-d_H\bth).
        \end{equation}
    \end{enumerate}
    
    Assume now that any of the equivalent conditions above is satisfied. Then
    \begin{enumerate}[label=(\roman*)]
        \item $\bL $ satisfies 
         \[
         \frac12\{\bL ,\bL \}^S = d_H \bL,
         \]    
         and it is a Maurer--Cartan element in $(C_\ham^\bullet,D-\LQ,\{\cdot,\cdot\}^B)$.
         \item If $(\bul{\wt{L}},Q,\bul{\wt{\theta}})$ is another Hamiltonian triple w.r.t.\ a $(\LQ-d_H)$-closed development $\bom$, such that $\wt{\theta}^0 = h_V\omega = \theta^0$, there exist $\bF\in\oloc^{0,<\top}$, $\bul{\gamma}\in\oloc^{1,<\top-1}$ and $K\in \Omega^\bullet(M)$ such that
        \[
        \bul{\wt{L}} - \bL = d_H\bF + p^*K, \quad \bul{\wt{\theta}} - \bth=d_V\bF + d_H\bul{\gamma}.
        \]
        In particular, $\bbP (\bul{\wt{L}} - \bL )=0$, and $\bom= -d_V\bul{\wt{\theta}} + d_H\bul{\beta}$ for $\bul{\beta}\in\oloc^{1,<\top}$.
        \item 
        Consider $\bul{\wt{\omega}} = \bom + (d_H - \LQ)\bul{\beta}$ for $\beta\in\oloc^{1,\bullet}$ of partial effective degree $\mathsf{ped}(\bul{\beta}) = k-1$ such that $\Pi(d_V\bul{\beta}_0)=0$. Then $(\bul{\wt{L}},Q,\bul{\wt{\theta}})$ is a Hamiltonian triple w.r.t. $\bul{\wt{\omega}}$ with
        \[
        \bul{\wt{\theta}} = \bth + h_V(d_H-\LQ)\bul{\beta}, \qquad \bul{\wt{L}} = L + P_V((d_V-d_H) e^{-\iota_Q}\bbbeta^\bullet),
        \]
        where
        \[
        \bbbeta^\bullet = \sum_{k=0}(\wt{H}_Q d_V)^k\bul{\beta}.
        \]
    \end{enumerate}
    
\end{theorem}

\begin{proof}
   Let us begin proving $(1)\implies (2)$. Since $\bom$ is $d_V$-closed, it is also $d_V$-exact:
    \begin{align*}
        \bom=d_V h_V \bom \leadsto \bth\doteq h_V\bom.
    \end{align*}
    Moreover, since $\bom$ is $(d_H-\LQ)$-closed, we get
    \begin{align*}
        d_V(\iota_Q \bom -d_H\bth)=
        \LQ\bom+d_Hd_V\bth= - \LQ\bom + d_H\bom=0.
    \end{align*}
    Using the homotopy equation for the vertical differential, $\id=[h_V,d_V] + p^*0^*$, we get 
    \begin{align*}
        \iota_Q\bom-d_H\bth=d_V h_V(\iota_Q\bom-d_H\bth).
    \end{align*}
    This means that, for $\bL =h_V(\iota_Q\bom-d_H\bth)$ and $\bth=h_V\bom$, we have that $(\bL , Q, \bth)$ is a Hamiltonian triple.

    Let us now prove that $(2)\implies (1)$. The  requirement $\bth=h_V \bom$ and the Hamiltonian triple condition, together, imply
    \[
    \LQ\bom = - d_V \iota_Q \bom = -d_Vd_H\bth = d_H d_V \bth = d_H \bom.
    \]

    Assume for the rest of the proof that $(2)$ (and thus $(1)$) holds. 
    
    (i) From $[Q,Q]=0$ we compute
    \[
    0= \iota_{[Q,Q]}\bom = 2 \iota_Q d_V \iota_Q\bom - d_V \iota_Q\iota_Q \bom = 2\iota_Q d_V d_H \bth - d_V \iota_Q\iota_Q \bom 
    \]
    whence, anticommuting various graded operators,
    \[
    \frac12 d_V \iota_Q\iota_Q \bom = d_H \iota_Q \bom = d_H d_V \bL .
    \]
    Now, applying the vertical homotopy equation we conclude
    \[
    \frac12 \iota_Q\iota_Q \bom = d_H\bL  - p^*0^*\left(\frac12 \iota_Q\iota_Q \bom-d_H\bL \right)
    \]
    whence we get $d_H\bL = \{\bL,\bL\}^S$ by observing that $0^*\frac12 \iota_Q\iota_Q \bom= 0^*d_H\bL =0$. Now we can simply evaluate the MC condition on $\bL$ to get
    \[
    \frac12\{\bL,\bL\}^B + (d_H-\LQ)\bL = \LQ(\bL) - \frac12\iota_Q\iota_Q\bom +d_H\bL - \LQ(\bL) = 0
    \]
    owing to the previously proven identity.
    
    (ii) Suppose, furthermore, that $(\bul{\wt{L}},Q,\bul{\wt{\theta}})$ is another Hamiltonian triple w.r.t.\ $\bom$, with $d_V\wt{\theta}^0 = \omega = d_V\theta^0$. Then we have that
    \[
    d_H\bom = \LQ\bom = \begin{cases}
        -d_V d_H\bth = d_H d_V\bth\\
        -d_V d_H\bul{\wt{\theta}} = d_H d_V\bul{\wt{\theta}}
    \end{cases}
    \]
    whence, since $d_V\bth = \bom$
    \[
    d_H(\bom - d_V\bul{\wt{\theta}})=0,
    \]
    applying $\id=[d_H,h^\nabla] + I\Pi$ to $\bom - d_V\bul{\wt{\theta}}$ we get
    \begin{equation}\label{e:wtthetacondition1}
    d_V\bul{\wt{\theta}}  = \bom - d_Hh^\nabla(\bom - d_V\bul{\wt{\theta}}) - I\Pi(\bom-d_V\bul{\wt{\theta}})=\bom - d_Hh^\nabla(\bom - d_V\bul{\wt{\theta}})
    \end{equation}
    since $\Pi(\bom-d_V\bul{\wt{\theta}})=0$ owing to $d_V\wt{\theta}^0 = \omega$ and $\Pi\vert_{\oloc^{\geq 1,<\top}}\equiv 0$.
    
    Now, observe that the Hamiltonian triple condition yields:
    \[
    d_V(\bul{\wt{L}} - \bL ) + d_H(\bul{\wt{\theta}} - \bth) = 0 \implies \Pi d_V(\bul{\wt{L}} - \bL ) = 0,
    \]
    whence $\bul{\wt{L}} - \bL  = p^*K + d_H\bF$ for some $K\in \Omega^\bullet(M)$ and $\bF\in\oloc^{0,<\top}$. In particular, $\bbP (\bul{\wt{L}} - \bL )=0$. Then, reinserting in the previous equation
    \[
    d_H(\bul{\wt{\theta}} - \bth - d_V \bF) = 0, \quad \forall \bullet<\top
    \]
    and from the homotopy equation
    \begin{align}\notag
    \bul{\wt{\theta}} - \bth 
    &= d_V \bF + d_Hh^\nabla(\bul{\wt{\theta}} - \bth - d_V \bF) + I\Pi (\bul{\wt{\theta}} - \bth - d_V \bF)\\\label{e:wtthetacondition2}
    &= d_V \bF + d_Hh^\nabla(\bul{\wt{\theta}} - \bth - d_V \bF) , \quad \forall \bullet<\top\\\notag
    &=d_V \bF + d_H\bul{\gamma}
    \end{align}
    where we used that $\Pi (\bul{\wt{\theta}} - \bth - d_V \bF)=0$ since $\bullet<\top$.

    (iii)  Let $\bul{\beta}_0\equiv\bul{\beta}\in\oloc^{2,\bullet}$ and consider $\bul{\wt{\omega}} = \bom + (d_H - \LQ)\bul{\beta}_0$. To enforce $d_V\bul{\wt{\omega}}=0$, from the homotopy associated to the retract in \eqref{e:LQdefretract} we have
    \[
    d_V\bul{\beta}_0 = (d_H - \LQ)\wt{H}_Q d_V\bul{\beta}_0 + \wt{H}_Q(d_H - \LQ)d_V\bul{\beta}_0 + \wt{I}\Pi d_V\bul{\beta}_0 = (d_H - \LQ)\wt{H}_Q d_V\bul{\beta}_0
    \]
    and we conclude $\bul{\beta}_0 = (d_H-\LQ)\bul{\beta}_1$, which also implies $d_V\bul{\beta}_2 = (d_H - \LQ)\bul{\beta}_3$ and so on, requiring that $\bul{\beta}_i = \wt{H}_Q d_V \bul\beta_{i-1}$. Since $(d_H-\LQ)\bul{\wt{\omega}} = 0$, by $(2)$ there exists a Hamiltonian triple $(\bul{\wt{L}},Q,\bul{\wt{\theta}})$ with
    \begin{align*}
    \bul{\wt{\theta}} = h_V\bul{\wt{\omega}} 
        &= \bth + h_V(d_H - \LQ)\bul{\beta}_0 
    \end{align*}
    and
    \begin{align*}
    \bul{\wt{L}} &= h_V\left(\iota_Q\bul{\wt{\omega}} - d_H \bul{\wt{\theta}}\right) \\
    &= h_V\left(\iota_Q\bom - d_H\bth\right) + h_V\iota_Q(d_H - \LQ)\bul{\beta}_0 + d_Hh_Vh_V(d_H - \LQ)\bul{\beta}_0 \\
    &= \bL+ h_V\iota_Q(d_H - \LQ)\bul{\beta}_0
    \end{align*}
    where we used that $h_V^2=0$. 

    Using the graded Cartan Calculus we can easily see that, since $0\equiv \iota_{[Q,Q]} = \iota_Q \LQ - \LQ \iota_Q$, for any local form $X\in\oloc^{\bullet,\bullet}$
    \[
    \iota_Q^k\LQ X = \frac{1}{k+1}\left(\iota_Q^{k+1} d_V - d_V\iota_Q^{k+1}\right)X.
    \]
    Hence, let's look at the term $\iota_Q(d_H - \LQ)\bul{\beta}_0$. We compute (recall  $[\iota_Q,d_H] = \iota_Q d_H - d_H \iota_Q$ for $Q$ evolutionary and odd)
    \begin{align*}
        \iota_Q(d_H - \LQ)\bul{\beta}_0 &= d_H \iota_Q \bul{\beta}_0 - \frac12 \iota_Q^2 d_V\bul{\beta}_0 + \frac12 d_V \iota_Q^2 \bul{\beta}_0\\
        &=d_H \iota_Q \bul{\beta}_0+ \frac12 d_V \iota_Q^2 \bul{\beta}_0 - \frac12 \iota_Q^2 (d_H-\LQ)\bul{\beta}_1\\
        &=d_H \left(\iota_Q \bul{\beta}_0  -\frac12 \iota_Q^2\bul{\beta}_1\right) + \frac12 d_V \iota_Q^2 \bul{\beta}_0 + \frac12 \iota_Q^2\LQ\bul{\beta}_1\\
        &=d_H \left(\iota_Q \bul{\beta}_0  -\frac12 \iota_Q^2\bul{\beta}_1\right) + \frac12 d_V \iota_Q^2 \bul{\beta}_0 + \frac16\iota_Q^3 d_V\bul{\beta}_1 - \frac16d_V\iota_Q^3 \bul{\beta}_1\\
        &=d_H \left(\iota_Q \bul{\beta}_0  -\frac12 \iota_Q^2\bul{\beta}_1\right) + d_V\left(\frac12 \iota_Q^2 \bul{\beta}_0 - \frac16 \iota_Q^3\bul{\beta}_1\right) + \frac16\iota_Q^3(d_H - \LQ)\bul{\beta}_2
    \end{align*}

    More generally
    \begin{align*}
    \frac{1}{k!}\iota_Q^k(d_H - \LQ)\bul{\beta}_{k-1} 
    &= \frac{1}{k!}d_H\iota_Q^k\bul{\beta}_{k-1} - \frac{1}{k!}\iota_Q^k\LQ\bul{\beta}_{k-1} \\
    &= \frac{1}{k!}d_H\iota_Q^k\bul{\beta}_{k-1} - \frac{1}{k+1!}\iota_Q^{k+1} d_V\bul{\beta}_{k-1} + \frac{1}{k+1!}d_V\iota_Q^{k+1}\bul{\beta}_{k-1}\\
    &=- \frac{1}{k+1!}\iota_Q^{k+1} (d_H-\LQ)\bul{\beta}_{k} + \frac{1}{k!}d_H\iota_Q^k\bul{\beta}_{k-1} + \frac{1}{k+1!}d_V\iota_Q^{k+1}\bul{\beta}_{k-1}
    \end{align*}
    We thus conclude that
    \[
    \iota_Q(d_H - \LQ)\bul{\beta}_0 = - d_H\left[ e^{-\iota_Q}\bul{\bbbeta}\right]^{1,\bullet} + d_V \left[ e^{-\iota_Q}\bul{\bbbeta}\right]^{0,\bullet} 
    \]
    where $\bul{\bbbeta} = \sum_k \bul{\beta}_k = \sum_k(\wt{H}_Qd_V)^h \bul{\beta}_0$. Hence, recalling that $P_V = h_V p^{1,\bullet}$
    \[
    \bul{\wt{L}} = \bL + h_V\left(- d_H\left[ e^{-\iota_Q}\bul{\bbbeta}\right]^{1,\bullet} + d_V \left[ e^{-\iota_Q}\bul{\bbbeta}\right]^{0,\bullet} \right) = \bL + P_V\left( (d_V-d_H)\left[ e^{-\iota_Q}\bul{\bbbeta}\right]\right)
    \]

\end{proof}

\begin{definition}\label{def:tripleredef}
    Let $\bom$ be a $(d_H-\LQ)$-closed symplectic development of a local $k$-symplectic form $\omega$, with $Q\in\Xloc^1$ cohomological, and a choice of homotopies $h^\nabla$ and $h_V$. The \emph{canonical Hamiltonian triple} associated to $(\bom,h^\nabla,h_V)$ is the data $(\bL , Q, \bth)\in \mathbb{HT}_{\bom}$ defined in Equation \eqref{e:canonicaltriple} of Theorem \ref{thm:Hamtriples}.
    
    The \emph{Noether Lagrangian} and the \emph{total Lagrangian} associated associated to a Hamiltonian triple $(\bL , Q, \bth)$ are, respectively
    \[\begin{cases}
        \bD\colon \mathbb{HT}_{\bom}\to \ohor^\bullet, & \bD(\bL , Q, \bth)\equiv\bul{\bD}\doteq \bL -\iota_Q\bth,\\
        \bbL\colon \mathbb{HT}_{\bom}\to \ohor^\bullet, & \bbL(\bL , Q, \bth)\equiv\bul{\bbL}\doteq \bL  + \LE\bul{\bD}
    \end{cases}
    \]
    Furthermore, consider the following:
    \begin{enumerate}
        \item Let $\bul{f}\in\oloc^{0,<\top}$. We call \emph{triple redefinition} the map
            \[
            \mathsf{T}_{\bul{f}} \colon \mathbb{HT}_{\bom}\to \mathbb{HT}_{\bom}, 
            \quad (\bF,\bb{\bF},\bth_{\bF}) \mapsto (\bF + d_H\bul{f},\bb{\bF},\bth_{\bF} + d_V\bul{f}).
            \]
        \item Let $\bul{\beta}\in\oloc^{2,\bullet}$ and $\bul{\wt{\omega}}=\bom + (d_H - \LQ)\bul{\beta}$. We call \emph{global redefinition} the map
            \[
            \mathsf{G}_{\bul{\beta}}\colon \mathbb{HT}_{\bom}\to \mathbb{HT}_{\bul{\wt{\omega}}}, \qquad (\bF + P_V\left( (d_V-d_H)\left[ e^{-\iota_Q}\bul{\bbbeta}\right]\right), \bb{\bF}, \bth_{\bF} + h_V(d_H - \LQ)\bul{\beta}_0)
            \]
    \end{enumerate}
\end{definition}

\begin{remark}
    The triple redefinition map tells us how one can change the canonical Hamiltonian triple, keeping the symplectic development $\bom$ fixed. Indeed $\mathsf{T}_{\bul{f}}(\bL,\bb{\bL},\bth_{\bL})$ explores all possible \emph{other} Hamiltonian triples, given the canonical $(\bL,\bb{\bL},\bth_{\bL})$, due to Theorem \ref{thm:Hamtriples}, part (ii). The global redefinition, instead, tells us how to change the data $(\bom,\bL,\bth)$ keeping just $(Q,\omega)$ fixed, owing to Theorem \ref{thm:Hamtriples}, part (iii).
\end{remark}

\begin{theorem}\label{thm:MCelement}
    The Noether and total Lagrangians $\bul{\bD},\ \bul{\mathbb{L}}$ associated to the Hamiltonian triple $(\bL , Q, \bth)$ are $(d_H-\LQ)$-cocycles in $\ohor^{0,\top}$ and $\mathbb{d}\doteq (0,\bul{\bD}),\ \mathbb{l}=(0,\bul{\mathbb{L}})$ are $(D-\LQ)$-cocycles in $C^\bullet_{\ham}$, and their cohomology classes are invariant under triple redefinition:
    \begin{align*}
    \bD(\mathsf{T}_{\bul{f}}(\bL,\bb{\bL},\bth_{\bL})) = \bD(\bL,\bb{\bL},\bth_{\bL}) + (d_H - \LQ)\bul{f}\\
    \bbL(\mathsf{T}_{\bul{f}}(\bL,\bb{\bL},\bth_{\bL})) = \bbL(\bL,\bb{\bL},\bth_{\bL}) + (d_H - \LQ)(1+\LE)\bul{f}
    \end{align*}
    
    Moreover, there is a Maurer--Cartan element in $\mathsf{MC}(C^\bullet_{\ham},D-\LQ,\{\cdot,\cdot\}^S,\{\cdot,\cdot,\cdot\}^S)$
    \[
    \mathbb{s}\doteq \mathbb{l}-\bbalpha
    \]
    where 
    \[
    \bbalpha\doteq \frac12\wt{H}_Q\left(\{\mathbb{l},\mathbb{l}\}\right),
    \]
    such that $\bbP \mathbb{s} = \bbP (0,\bL)$ is a Maurer--Cartan element in $\left(\calF ,d_\ham,\{\cdot,\cdot\}_\ham\right)$. We call $\mathbb{s}$ the \emph{standard} Maurer--Cartan element of the Hamiltonian triple $(\bL , Q, \bth)$.
\end{theorem}

\begin{proof}
    The fact that $\bul{\bD}$ and $\bul{\mathbb{L}}$ are $d_H-\LQ$ cocycles in $\ohor^{0,\bullet}$ was proven in \cite[Theorem 23]{MSW}. We report the argument:
    From the Hamiltonian triple equation we have
    \[
    d_V (\bL  -\iota_Q\bth) = (\LQ - d_H)\bth \iff d_V \bul{\bD} = (\LQ - d_H)\bth
    \]
    whence
    \[
    d_V(\LQ - d_H)\bul{\bD} = 0
    \]
    but $\mathsf{ped}((\LQ - d_H)\bul{\bD}) = 1$ and since there are no nontrivial nonzero-degree constants, it must be $(\LQ - d_H)\bul{\bD}=0$.
    Then $(D-\LQ)\mathbb{d}=(D-\LQ)\mathbb{l}=0$ follows directly from $(D-\LQ)(0,X) = (0,(d_H-\LQ)X)$ for any $X$.
    
    The proof that the $(d_H-\LQ)$-classes of $\bul{\bD}$ and $\bul{\bbL}$ are invariant under the triple redefinition was also given in \cite[Proposition 27]{MSW}, using different terminology.
    
    Then we observe that
    \[
    (D-\LQ)\{\mathbb{l},\mathbb{l}\}^S = \left(0,(d_H-\LQ)\{\bul{\mathbb{L}},\bul{\mathbb{L}}\}^S\right) =\left(0, 2\{(d_H-\LQ)\bul{\mathbb{L}},\bul{\mathbb{L}}\}\right) = (0,0) 
    \]
    and $\bbP \{\mathbb{l},\mathbb{l}\}^S =0$ so that, from the homotopy equation $\id = [\wt{H}_Q,D-\LQ] + \wt{I}\bbP $ given by the $\LQ$-perturbed deformation retract \eqref{e:LQdefretract},
    \[
    \frac12\{\mathbb{l},\mathbb{l}\}^S = (D-\LQ)\bbalpha, \qquad \bbalpha\doteq \frac12\wt{H}_Q\{(0,\bul{\mathbb{L}}),(0,\bul{\mathbb{L}})\}^S,
    \]
    and thus, recalling that the three bracket vanishes on functionals such that $\{\bL,\bL\}^S = d_H(\dots)$, the standard object $\mathbb{s}\doteq \mathbb{l} - \bbalpha$ is a Maurer--Cartan element:
    \[
    (D-\LQ) \mathbb{s} + \frac12 \{\mathbb{s},\mathbb{s}\}^S + \frac16\{\mathbb{s},\mathbb{s},\mathbb{s}\}^S = 0.
    \]
\end{proof}

\begin{corollary}
    Let $\mathbb{l}=(0,\bul{\bbL})$ as in Theorem \ref{thm:MCelement} and consider a Maurer--Cartan element of the form $\ell=\bbP \mathbb{l}\in\mathsf{MC}(\calF,d_{\ham},\{\cdot,\cdot\}_{\ham})$. Its image under the quasi inverse $\wt{\mathbb{I}}^S\colon (\calF_\ham, d_\ham, \{\cdot,\cdot\}_\ham) \to (C^\bullet_\ham, D-\LQ,\{\cdot,\cdot\}^S,\{\cdot,\cdot,\cdot\}^S)$ is 
    \[
    \wt{\mathbb{I}}_\mathrm{MC}(\ell) = \wt{i}(\ell) - \frac12\wt{H}_Q\{\wt{i}(\ell),\wt{i}(\ell)\}^S 
    = \mathbb{s} + \mathbb{r}
    \]
    where
    \[
    \wt{i}(\ell) = (0,\ell) + \sum_{k\geq 1}(H\LQ)^k (0,\ell), 
    \]
    $\mathbb{s}\in\mathsf{MC}(C_\ham^\bullet,D-\LQ, \{\cdot,\cdot\}^S,\{\cdot,\cdot,\cdot\}^S)$ is as in Theorem \ref{thm:MCelement}, and
    \[
    \mathbb{r} = \wt{i}\bbP\left(0,\bul{\bbL}\right) - (0,\bul{\bbL}) = -[\wt{H}_Q,D-\LQ]\mathbb{l}, \qquad \bbP\mathbb{r} = 0.
    \]
\end{corollary}
\begin{proof}
    Recall that the map induced by an $L_\infty$-morphism is defined on Maurer--Cartan elements
    by 
        \begin{align*}
            \wt{\mathbb{I}}_\mathrm{MC}(\ell)=
            \sum_{k=1}^\infty \frac{1}{k!}\mathbb{I}_k^1(\ell^{\vee k})
        \end{align*}
        
    Using the explicit formula for $\mathbb{I}_k^1$ from Lemma \ref{lem:quasiinverse}, we see that for $k\geq 2$
        \begin{align}
           \label{Eq: BigBeautifulformula}
            \mathbb{I}_k^1 
            = - \wt{H}_Q(
            \delta^1_2 \mathbb{I}_{k}^2 + \delta^1_3 \mathbb{I}_{k}^3 - \mathbb{I}_{k-1}^2\delta^{k-1}_{\mathcal{F}k}),
        \end{align}
        since for all the other terms the coderivations corresponidng to the $L_\infty$-algebras vanish. 
        The term $\delta^{k-1}_{\mathcal{F},k}(\ell^{\vee {k}})=0$, since $\delta^{k-1}_{\mathcal{F},k}$ is induced by the bracket on $\mathcal{F}_\ham$ and $\{\ell,\ell\}_\ham=0$. Moreover, 
        if $k\geq 4$ also the first two summands vanish, since 
        $I^1_n$ lays in the image of $\wt{H}_Q$ for $n\geq 2$ on which the coderivations $\delta^1_2$ and $\delta_3^1$ vanish. 
        For $k=3$ the first argument vanishes by the same argument. We conclude 
        \begin{align*}
            \wt{\mathbb{I}}(\ell)
            = 
            \wt{i}(\ell) + \wt{\mathbb{I}}^1_2(\ell,\ell) + \wt{\mathbb{I}}^1_3(\ell,\ell,\ell)
        \end{align*}
    which gives by formula \eqref{Eq: BigBeautifulformula}
    \[
    \wt{\mathbb{I}}(\ell) = \wt{i}(\ell) - \wt{H}_Q\{\wt{i}(\ell),\wt{i}(\ell)\}^S \pm \wt{H}_Q\left( \mathrm{Jac}_{\{\cdot,\cdot\}^S}(\wt{i}(\ell)^{\otimes 3})\right).
    \]
    We observe that $\wt{i}(\ell)=\mathbb{p} \mathbb{l} + H(\dots)$ by the explicit formula of $\wt{i}$ (Lemma \ref{lem:pert.ham.def.ret}) and therefore, recalling that $\ell=(0,\bL)$  
    \[
    \wt{\mathbb{I}}(\ell) = \wt{i}(\ell) - \wt{H}_Q\left(0,\{\bL,\bL\}^S\right) \pm \wt{H}_Q\left(0, \mathrm{Jac}_{\{\cdot,\cdot\}^S}(\bL)^{\otimes 3}\right).
    \]
    since the bracket $\{\cdot,\cdot\}^S$ vanishes on the image of $H$. Finally, the Hamiltonian vector field of $\bL$ is $Q$, and the Hamiltonian vector field of $\{\ell, \ell\}^S$ is $[Q,Q]=0$. Thus, also the Jacobiator vanishes. 

    We conclude by observing that $\wt{H}_Q\left(0,\{\ell,\ell\}\right) = \bbalpha$ (Theorem \ref{thm:MCelement}, and $\wt{i}(\ell) = \wt{i} \bbP \mathbb{l} = \mathbb{l} - [\wt{H}_Q,D-\LQ]\mathbb{l}$, whence
    \[
    \wt{\mathbb{I}}_{\mathsf{MC}}(\ell) = (\wt{i}\bbP - \id)  \mathbb{l} + \mathbb{l} - \bbalpha = \mathbb{s} - [\wt{H}_Q,D-\LQ] \mathbb{l} \doteq \mathbb{s} - \mathbb{r},
    \]
    where, clearly, $\bbP \mathbb{r} = 0$.

\end{proof}

\begin{proposition}\label{propr:gauge-eqMC}
    Let $\bL$ and $\bul{\wt{L}}$ be two developments of $L$, such that $(\alpha, \bL)$ and 
    $(\beta,\bul{\wt{L}})$ are Maurer-Cartan elements in $(C^\bullet_{\ham},D-\LQ,\{\cdot,\cdot\}^S,\{\cdot,\cdot,\cdot\}^S)$, then they are $L_\infty$-gauge equivalent as Maurer-Cartan elements. 
\end{proposition}
\begin{proof}
 If $\bL$ and $\bul{\wt{L}}$ are two developments of $L$, then $\bbP (\alpha, \bL)=\bbP (\beta, \bul{\wt{L}})=\bbP(0,L)$ and therefore the Maurer Cartan elements induced by the $L_\infty$ morphism $\wt{I}\circ \bbP $ are the same, i.e.\ $\tilde{I}\circ\bbP _\mathsf{MC}((\alpha, \mathbb{L}^\bullet))=\tilde{I}\circ\bbP _\mathsf{MC}((\beta, \tilde{\mathbb{L}}^\bullet))$ and therefore, since $\bbP\circ I\sim \id$, we get that $(\alpha, \mathbb{L}^\bullet)$ and 
    $(\beta,\wt{\mathbb{L}}^\bullet)$ are $L_\infty$-gauge equvialent, see \cite[Prop 6.6]{kraft:2024}.
\end{proof}

\begin{corollary}
    Let $(\bul{L},Q,\bth)$ and $(\bul{\wt{L}},Q,\wt{\bth})$ be two hamiltonian triples. Then, their associated standard Maurer--Cartan elements $\mathbb{s},\wt{\mathbb{s}}$ are $L_\infty$-gauge equivalent, and both are $L_\infty$-gauge equivalent to $\wt{\mathbb{I}}_{\mathsf{MC}}(\bbP(0,L))$.
\end{corollary}

\begin{proof}
    We simply apply Proposition \ref{propr:gauge-eqMC} to the standard Maurer--Cartan elements $\mathbb{s}$ and $\wt{\mathbb{s}}$, which are both developments of $L$ in their Horizontal part of the Hamiltonian cone.
\end{proof}

\section{BV formalism revisited}\label{sec:BVBFV}

The BV formalism is a tool to study field theories with local symmetries, motivated by the problem of quantisation. It provides a description of the derived critical locus of a local functional on $\mathcal{E}$ in terms of auxiliary (graded) field theoretic data.

Originally, this was devised to discuss field theory on closed manifolds (or with vanishing falloff conditions for fields at infinity), and it was later extended to interact with higher codimension data like boundaries and corners. For our purposes this means working with inhomogeneous local forms in the horizontal direction, so with objects in $\oloc^{p,\bullet}(\mathcal{E}\times M)$. 

Our main references throughout will be \cite{CMR2014,CMR18,MSW}, as well as \cite{BarnichBrandtHenneaux,Grigoriev} and the earlier works \cite{Sharapov1,Sharapov2}. The idea of computing descent (i.e. special developments of local forms) traces back to \cite{ZUMINO1985477,ZuminoDescent}. Some earlier investigations along these lines appear in \cite{KotovStrobl}.

\subsection{Generalities}

In this section we will give a review of the main building blocks of the BV construction. 

Let, as above $\mathcal{E}=\Gamma(E\to M)$ denote the space of sections of a graded vector bundle.\footnote{This can be easily extended to bundles that are nonlinear in degree $0$.} We consider once again a local $k$-symplectic form on it.

\begin{definition}\label{laxBVBFV}
A \emph{BV theory} is the assignment, to a manifold $M$ of the space of sections $\mathcal{E}$ of a vector bundle $E\to M$, a $(-1)$-symplectic local form $\omega$ on $\mathcal{E}$, and a cohomological vector field $Q\in \Xloc(\mathcal{E})$ such that $\Pi(\LQ\omega)=0$, collectively denoted by $\fracF=(\mathcal{E},\omega,Q)$. A \emph{lax presentation}\footnote{This presentation is due to Cattaneo, Mnev and Reshetikhin \cite{CMR2014,CMR18}. See also \cite{MSW}.} of the BV theory is given by a $(d_H-\LQ)$-closed $k$-symplectic development $\bom$ of $\omega=\omega^0$, and a Hamiltonian triple $(\bL,Q,\bth)\in\mathbb{HT}_{\bom}$ with $\bom = d_V\bth$.

Given a BV theory, its associated BV local complex is the cochain complex 
    \[
    \left(\bul{\mathbb{\Omega}}_{\loc}(\mathcal{E}),d_H - \LQ\right), \qquad \bul{\mathbb{\Omega}}_{\loc}(\mathcal{E}) \doteq \bigoplus_{k=0}^{\mathrm{dim}(M)}\oloc^{\bullet,k}(\mathcal{E}\times M).
    \]
    The local cohomology will be denoted $\bul{\mathbb{H}}_{\loc}(\mathcal{E})$. A class $[\bul{\mathcal{O}}]_{\LQ-d_H}\in\bul{\mathbb{H}}_{\loc}(\mathcal{E})$ is called a local observable.

    Given a BV theory and a $(d_H-\LQ)$-closed development of $\omega$, its associated standard $L_\infty$ algebra is given by Theorem \ref{Thm: Linfty}
    \[
    \mathfrak{L}^S(\bom,Q)\doteq \left(C_\ham^\bullet, D-\LQ, \{\cdot,\cdot\}^S, \{\cdot,\cdot,\cdot\}^S\right),
    \]
    while the associated BV dgL$[k]$ algebra is 
    \[
    \mathfrak{L}^B(\bom,Q)\doteq \left(C_\ham^\bullet, D-\LQ, \{\cdot,\cdot\}^B\right),
    \]
    where in both cases $(C_\ham^\bullet,D)$ is the Hamiltonian cone of $\bom$ (Definition \ref{def:HamCone}).
\end{definition}

\begin{theorem}
    Let $(\bL,Q,\bth)\in\mathbb{HT}_{\bom}$ be a lax presentation of BV theory. Then:
    \begin{enumerate}
        \item $\bL$ satisfies the modified classical master equation
        \[
        d_H \bL = \frac12\{\bL,\bL\},
        \]
        and $(0,\bL)\in\mathsf{MC}\left(\mathfrak{L}^B(\bom,Q)\right)$.
        \item The Noether and total Lagrangians, respectively given by 
        \[
        \bul{\bD} = \bL - \iota_Q\bth; \quad \bul{\bbL} = \bL + \LE \bul{\bD}
        \]
        are cocycles of the local complex.
        \item The cohomology class of the Noether and total Lagrangians are invariant under the Hamiltonian triple redefinition map (Definition \ref{def:tripleredef}).
        \item The total Lagrangian determines a Maurer--Cartan element
        \[
        \mathbb{s}\in \mathsf{MC}(\mathfrak{L}(\bom,Q)), \qquad \mathbb{s} = (0,\bul{\bbL}) - \wt{H}_Q\left(\frac12\{(0,\bul{\bbL}),(0,\bul{\bbL})\}^S\right).
        \]
    \end{enumerate}
\end{theorem}

\begin{proof}
    This is an application of Theorem \ref{thm:Hamtriples}, where $(1)$ is proven, and Theorem \ref{thm:MCelement}, where (2)-(4) are proven.
\end{proof}

The cocycle condition in the local complex should be interpreted as the descent equations of \cite{ZuminoDescent} (see also \cite{ZUMINO1985477}). The fact that the classes of $\bul{\bD}$ and $\bul{\bbL}$ are invariant under the Hamiltonian triples redefinition map of Definition \ref{def:tripleredef} states that, for a chosen development $\bom$, a change of presentation does not affect certain observables.

In general, one could change the presentation of the theory via the (more general) global redefinition map of Definition \ref{def:tripleredef}, which also changes $\bom$. In general it is not clear that this will preserve the theory in a simple way, as changing $\bom$ results in changing the $L_\infty$ algebra as well. We can however introduce the following \emph{special} type of full redefinition, corresponding to $\bul{\beta}=d_V \bul{\eta}$. (A similar construction was investigated in \cite{Fossati}.)

\begin{proposition}
    Let $\bul{\wt{\omega}} = \bom + (d_H-\LQ)d_V\bul{\eta}$. There is a map of Hamiltonian triples $\Liou\colon \mathbb{HT}_{\bom}\to \mathbb{HT}_{\bul{\wt{\omega}}}$ defined by
    \[
    \Liou\colon (\bL, Q ,\bth) \longmapsto (\bul{\wt{L}},Q,\bul{\wt{\theta}})=(\bL  + \frac12\iota_Q\iota_Q d_V  \bul{\eta}, Q,\bth - d_H\bul{\eta} + \iota_Q d_V \bul{\eta}),
    \]
    which preserves the cohomology classes of $\bul{\bD}$ and $\bul{\bbL}$, and $d_V\bul{\wt{\theta}} = \bul{\wt{\omega}}$. We call the map $\mathsf{L}$ the Liouville redefinition of Hamiltonian triples.
\end{proposition}

\begin{proof}
    We first observe that
    
    \[
    d_V\bul{\wt{\theta}} = - d_Vd_H \bul{\eta} + d_V\iota_Qd_V\bul{\eta} = (d_H - \LQ)d_V\bul{\eta}.
    \]
    Then, compute
    \begin{align*}
        d_V\bul{\wt{L}} + d_H\bul{\wt{\theta}}&= d_V  (\bL  + \frac12\iota_Q\iota_Q d_V \bul[]{\eta}) + d_H(\bth - d_H\bul{\eta} + \iota_Q  d_V \bul[]{\eta}) \\
        &= d_V\bL + d_H \bth - \iota_Q\LQ d_V\bul{\eta} + d_H\iota_Q d_V\bul{\eta}\\
        &=\iota_Q\bom + \iota_Q\left(d_H-\LQ\right)d_V\bul{\eta}
    \end{align*}
    where we used that $\iota_Q\LQ d_V \bul[]{\eta} = - \frac12 d_V \iota_Q\iota_Q d_V \bul[]{\eta}$, the fact that $Q$ is evolutionary, i.e.\ $\iota_Q d = d\iota_Q$, and $d^2\bul[]{\eta} = 0$.
    \end{proof}

    \begin{remark}
        Observe that being the image of the Liouville redefinition $(\bul{\wt{L}},Q,\bul{\wt{\theta}})$ a Hamiltonian triple w.r.t.\ $\bul{\wt{\omega}}$, we have, in particular, that the modified classical master equation is also preserved 
        \[
        \frac12\iota_Q\iota_Q\bul{\wt{\omega}} = d_H \bul{\wt{L}},
        \]
        meaning that $(0,\wt{L}^\bullet)$ is a Maurer Cartan element of $(C_\ham^\bullet,D-\LQ,\wt{\{\cdot,\cdot\}}{}^B)$, where $\wt{\{\cdot,\cdot\}}{}^B$ is defined using $\wt{\omega}^\bullet$.
        The Liouville redefinition is a particular case of the general redefinition of Definition \ref{def:tripleredef}, for $\bul{\beta}=d_V\bul{\eta}$, which is the generic situation for $\bul{\beta}$-closed, owing to the acyclicity of the vertical complex.
    \end{remark}

\subsection{Multisymplectic data}\label{sec:multisymp}
We are going to repackage a lax presentation of a BV theory in a compact way, highlighting its relation with multisymplectic geometry. We thank Christian Blohmann for suggesting us to look in this direction, and for providing some helpful insight that helped find Theorem \ref{thm:multisympBV}. 

Note that similar formulas emerged from \cite{GetzlerCSformal}, although we are not yet aware of a direct link between the cited work and ours. 

Observe, finally, that it is known that on multisymplectic manifolds one can build an $L_\infty$ algebra, due to \cite{rogers}. That structure, applied in our scenario, would exploit the horizontal form degrees of the symplectic form $\omega$, while our $L_\infty$ algebra uses the vertical part of the symplectic form. We believe the two constructions should combine, eventually, and defer this investigation to future work. (See also \cite{BlohmannGR} on this.)

Recall that the total differential on local forms $\mathrm{Tot}(\oloc^{\bullet,\bullet}(\mathcal{E}\times M))$ is $d= d_V  + d_H$. 
\begin{definition}
    Let $(\bL ,\bth, Q)\in\mathbb{HT}_{\bom}$ be a triple presenting the lax BV-BFV theory $\fracF=(\omega,Q)$. The \emph{BV multisymplectic momentum map} is the local form of total effective degree $0$:
    \[
    \lambda^\bullet \doteq L^\bullet + \theta^\bullet.
    \]
\end{definition}

\begin{theorem}\label{thm:multisympBV}
    Let $\mathfrak{F}=(\mathcal{E},\omega,Q)$ be a BV theory. Then $(\bL ,\bth,Q)\in\mathbb{HT}_{\bom}$ is a lax presentation of $\mathfrak{F}$ if and only if there exists $\lambda^{\bullet,\bullet}\in\oloc^{\bullet,\bullet}$ such that
    \begin{equation}\label{e:BVBFVrepackage}
        e^{\iota_Q}\bom = d \lambda^{\bullet,\bullet},
    \end{equation}
    with $\lambda^{0,\bullet} = \bL$, $\lambda^{1,\bullet} =\bth$, and under a triple redefinition we have
    \[
    \mathsf{T}_{\bul{F}}(\lambda^{\bullet,\bullet}) = \lambda^{\bullet,\bullet} + d \bul{f}.
    \]
    Moreover
    \[
    (\LQ- d_H) \lambda^{\bullet,\bullet} = d\bul{\bD}.
    \]
\end{theorem}

\begin{proof}
    Assume that $(\bL,Q ,\bth)\in\mathbb{HT}_{\bom}$ is a lax presentation of a BV theory. We compute the left-hand side:
    \[
    e^{\iota_Q}\bom= \bom + \iota_Q\bom + \frac12 \iota_Q\iota_Q\bom.
    \]
    Using the hamiltonian triple condition and the fact that $\bom= d_V \bth$ we get
    \[
    e^{\iota_Q}\bom =  d_V  \bth +  d_V  \bL  + d_H\bth + d_H\bL  = d_H\lambda^{\bullet,\bullet}.
    \]
    On the other hand, $e^{\iota_Q}\bom= \bom + \iota_Q\bom + \frac12 \iota_Q\iota_Q\bom$ is inhomogeneous in the vertical form degree, so that, from $e^{\iota_Q}\bom=d_H\lambda^{\bullet,\bullet}$ we obtain the axioms of a lax presentation of a BV theory 
    \begin{equation}
        \begin{cases}
            \bom =  d_V  \bth,\\
            \iota_Q\bom =  d_V \bL  + d_H\bth,\\
            \frac12\iota_Q\iota_Q\bom = d_H\bL 
        \end{cases}
    \end{equation}
    Finally, the last statement follows from the straightforward calculation: 
    \begin{align*}
    (\LQ-d_H) \lambda^{\bullet,\bullet} &= (\LQ-d_H) (\bL  + \bth) \\
    &= (\LQ-d_H) \bL  + (\LQ-d_H) \bth = d_H\bul{\bD} +  d_V \bul{\bD} = d\bul{\bD},
    \end{align*}
\end{proof}

\begin{corollary}\label{cor:twistedmultisymplectic}
The following relation holds for any lax presentation of a BV theory:
\begin{equation}
    \bom = (d - \LQ) e^{-\iota_Q}\lambda^{\bullet,\bullet} = (d-\LQ)(\mathbb{\Delta}^\bullet + \bul{\theta}).
\end{equation}
Hence $e^{-\iota_Q}\lambda^{\bullet,\bullet}=\mathbb{\Delta}^\bullet + \bul{\theta}$ is a primitive of the $(d-\LQ)$-closed symplectic form $\bom$ in $\ohor$.
\end{corollary}
\begin{proof}
    This is a direct consequence of Theorem \ref{thm:multisympBV}, since we can write
    \[
    \bom = e^{-\iota_Q} d \lambda^{\bullet,\bullet} = e^{-\iota_Q} d e^{\iota_Q} e^{-\iota_Q} \lambda^{\bullet,\bullet} = (d - \LQ) e^{-\iota_Q} \lambda^{\bullet,\bullet}=(d-\LQ)(\mathbb{\Delta}^\bullet + \bul{\theta}).
    \]
    This formula can also be checked easily by direct inspection.
\end{proof}

\begin{remark}Corollary \ref{cor:twistedmultisymplectic} can be interpreted in terms of shifted symplectic forms on the $k$-BV cohomology in $(\calF_\ham,d_\ham)$, i.e. as building a primitive for the lift to $C_\ham^\bullet$ of the shifted symplectic form on cohomology.
\end{remark}

    Note that when collapsing the BV-BFV data to $\mathsf{ghd}=0$ and restricting to the body of the graded manifold, Corollary \ref{cor:twistedmultisymplectic} reduces to 
    \[
    \omega^{0}\vert_{\mathsf{Body}}=d\lambda^{0}\vert_{\mathsf{Body}} =  d_V  L^{0}\vert_{\mathsf{Body}} + d\theta^{0}\vert_{\mathsf{Body}} +  d_V  \theta^{0}\vert_{\mathsf{Body}} = EL +  d_V \theta^{0}\vert_{\mathsf{Body}}, 
    \]
    where $EL$ is the (source) form obtained by $ d_V  L^{(0)}\vert_{\mathsf{Body}} = \Pi d_VL^{(0)}\vert_{\mathsf{Body}} - d\theta^{(0)}\vert_{\mathsf{Body}}$. Thus, the degree-$0$ part of the BV multisymplectic momentum map (restricted to the body) is the Lepage form associated to the classical Lagrangian $L^{0}$.

    We recall that a homotopy momentum map \cite{CalliesFregierRogersZambon} is given when one has an $L_\infty$-morphism between a Lie algebra $\fg$ and the $L_\infty$ algebra on a presymplectic manifold $(M,\omega)$. It is known that the solution space $(\mathcal{EL}=\{EL=0\},\omega_{\mathcal{EL}})$ is a presymplectic manifold (see \cite{RielloSchiavina_PS} and therein), and that the BV data encodes a Lie algebra action (at least) on the solution space \cite{BarnichBrandtHenneaux}. Moreover, there are examples where one can directly find a homotopy momentum map for the Lagrangian field theory by employing multisymplectic arguments on the variational bicomplex \cite{BlohmannGR,Bernardy_homotopy}.
    
    This leads to the following:

    \begin{conjecture}\label{Conjecture}
        Out of the data of a Hamiltonian triple $(\bul{L},\bul{\theta},Q)$ one can extract a homotopy momentum map on $(\mathcal{EL},\omega_{\mathcal{EL}})$, in the sense that there is a strict homotopy momentum map on the cohomology of the Koszul complex associated to $\mathcal{EL}$, if the Hamiltonian triple yields a Lax presentation of $\mathfrak{F}$ (Theorem \ref{thm:multisympBV}).
    \end{conjecture}

    \begin{remark}
        We stress that in order to realise Conjecture \ref{Conjecture} one might need to 
        \emph{lift} the axioms of a homotopy momentum map from the cohomology in degree $0$ of the Koszul(-Tate) complex to an $L_\infty$-algebra (bigraded by the Koszul--Tate directions as well) so that they can hold \emph{up to homotopy}. In other, slightly ridiculous, terms, one is looking for a homotopy-homotopy moment map. 
    \end{remark}

\begingroup
\sloppy
\printbibliography
\endgroup

\end{document}